\documentclass{article}

\usepackage{arxiv}

\usepackage[utf8]{inputenc} % allow utf-8 input
\usepackage[T1]{fontenc}    % use 8-bit T1 fonts
\usepackage{hyperref}       % hyperlinks
\usepackage{url}            % simple URL typesetting
\usepackage{booktabs}       % professional-quality tables
\usepackage{amsfonts}       % blackboard math symbols
\usepackage{nicefrac}       % compact symbols for 1/2, etc.
\usepackage{microtype}      % microtypography
\usepackage{lipsum}		% Can be removed after putting your text content
\usepackage{graphicx}

\input{mydefs.tex}
\usepackage{amssymb,amsmath}
\usepackage{multirow}
\usepackage{eurosym}
\usepackage{natbib}

\title{Knapsack Voting for Participatory Budgeting\thanks{This paper appeared in the ACM Transactions on Economics and Computation, July 2019 (\url{https://doi.org/10.1145/3340230).}}}

\author{
    Ashish Goel\\
	Stanford University\\
	\texttt{ashishg@stanford.edu} \\
	\And
	Anilesh K. Krishnaswamy\\
	Stanford University\\
	\texttt{anilesh@stanford.edu}\\
	\And
	Sukolsak Sakshuwong\\
	Stanford University\\
	\texttt{sukolsak@stanford.edu}\\
	\And
	Tanja Aitamurto\\
	Stanford University\\
	\texttt{tanjaa@stanford.edu}
}

\date{}

\renewcommand{\headeright}{Goel et al.}
\renewcommand{\undertitle}{}
\renewcommand{\shorttitle}{Knapsack Voting}

\hypersetup{
pdftitle={Knapsack Voting},
}

\begin{document}
\maketitle

\begin{abstract}
We address the question of aggregating the preferences of voters in the context of participatory budgeting. We scrutinize the voting method currently used in practice, underline its drawbacks, and introduce a novel scheme tailored to this setting, which we call ``Knapsack Voting''. We study its strategic properties - we show that it is strategy-proof under a natural model of utility (a dis-utility given by the $\ell_1$ distance between the outcome and the true preference of the voter), and ``partially" strategy-proof under general additive utilities. We extend Knapsack Voting to more general settings with revenues, deficits or surpluses, and prove a similar strategy-proofness result. To further demonstrate the applicability of our scheme, we discuss its implementation on the digital voting platform that we have deployed in partnership with the local government bodies in many cities across the nation. From voting data thus collected, we present empirical evidence that Knapsack Voting works well in practice.
\end{abstract}

% keywords can be removed

\section{Introduction}\label{sec:intro}
Direct democracy has gained a lot of significance lately, with many novel attempts at engaging citizens very directly in policy-making \citep{pateman2012participatory,smith2009democratic}. An exciting new development in this space is participatory budgeting \citep{cabannes2004participatory}, in which a local government body asks residents to vote on project proposals to decide how they should allocate their budgetary spending. The proposals could be, in a particular city/ward for instance, resurfacing streets, adding street lights, building playgrounds for children or renovating recreational facilities like parks.
	Participatory budgeting has had a long history in South America \citep{schneider2002budgets}, where it was born out of a need to address inequalities through democratic reform. It is now gaining popularity in the US, with cities like San Francisco, Vallejo, Boston, Chicago and New York adopting this paradigm \citep{pbp}. This development necessitates a detailed look at the voting method used in current ballots, and motivates the following question: when there are projects with different costs, and a fixed budget,  how can the varied preferences of voters be best aggregated? We address this question by proposing Knapsack Voting, and discuss its advantages over existing methods. We show that it has desirable strategic properties, and provide empirical evidence that it works well in practice. We have been successful in implementing Knapsack Voting as the official ballot procedure in the Youth Lead the Change PB election, 2016, in Boston (\url{https://youth.boston.gov/youth-lead-the-change/}). Independently of our work, other implementations of the same method have been used in some elections in Europe (see Section \ref{subsec:otherindep} for a discussion) -- adding to the relevance of our work in studying the advantages of Knapsack Voting.
	%Madrid's PB elections since 2016 (\url{https://decide.madrid.es/}), adding to the importance of our work.
	Overall, this piece of work is a step in the direction of addressing the need for mechanisms to facilitate complex decision-making processes. In designing the Knapsack Voting method, we incorporate ideas from the classical Knapsack Problem to make the voter choose projects under the budget constraint. In doing so, we align the constraints on the voters' decisions, to those of the decision maker's. We believe that this approach of aligning incentives is useful more broadly in designing mechanisms for complex participatory decision-making processes. Moreover, we want our schemes to be amenable to implementation as convenient digital interfaces to help individuals make informed budgeting decisions.

\subsection{Participatory budgeting in practice} \label{sec:kapproval}
PB first began in Porto Alegre in Brazil in 1988, and since then has become increasingly popular in Brazil\footnote{\url{https://siteresources.worldbank.org/INTEMPOWERMENT/Resources/14657_Partic-Budg-Brazil-web.pdf}}. Over the past couple of decades, PB has spread to Europe and North America, and more recently taken root in Asia and Africa \citep{hopefordem}. As of 2015, there have been nearly 1500 instances of PB around the world \citep{ganuza2012power}.

The first PB process in the United States was initiated in the 49th ward of Chicago in 2009. Since then PB has spread to many other cities like Vallejo, New York, Cambridge, to name a few. These elections traditionally used only paper ballots, and a voting method known as K-approval that we will describe shortly. 

\subsubsection{Our digital voting platform}
To popularize digital voting, we built a digital voting platform (\url{https://pbstanford.org/}) for the PB process of Chicago's 49th ward in 2012. We have since customized this tool to be easily adaptable to other elections. Many of the cities/municipalities doing PB have taken to digital voting by adopting our platform, while some cities like Long Beach and Dieppe have used our platform for their very first attempt at PB. In all, our platform has been used in over 25 different elections in around a dozen cities/wards over the past few years. One thing to note is that digital voting here means voting in person at a polling booth but via the digital platform. Most of these elections are required to have a paper ballot for those voters that opt out of digital voting. Some elections, for example PB Cambridge, implemented internet voting, where voters can authenticate themselves and vote from anywhere, via a voting website designed using our platform. These elections had accompanying paper ballots also.

\begin{figure}
\centering
  \includegraphics[width=240pt]{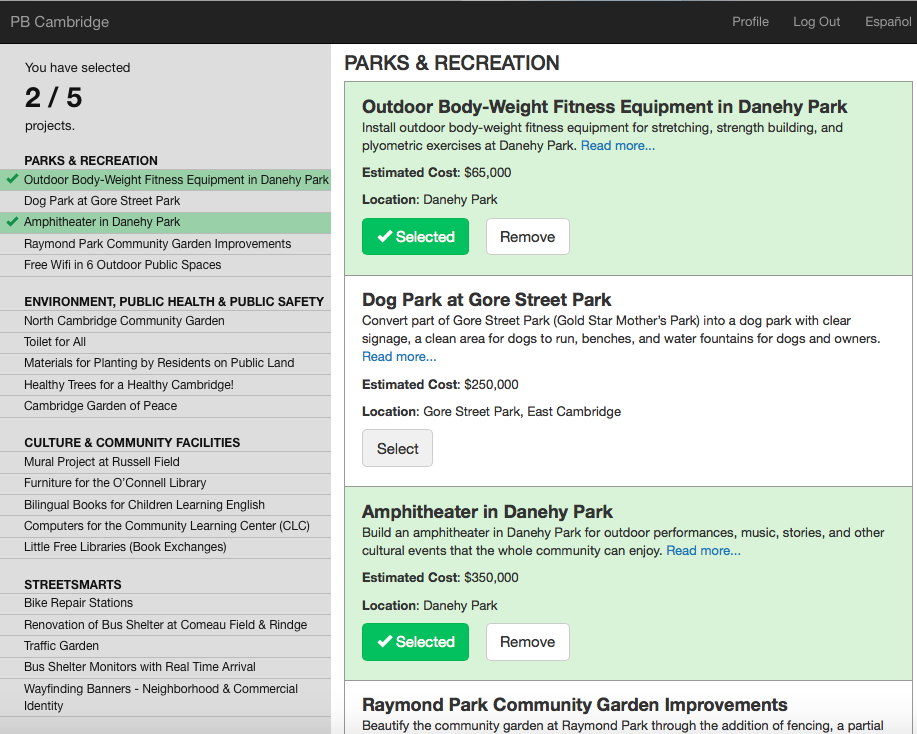}
    \caption{Our interface for K-approval: example from PB Cambridge 2015}
  \label{fig:kapproval}
\end{figure}

\subsubsection{K-approval voting}
The voting method currently used by most participatory budgeting elections is  \emph{K-approval} voting. Our interface for this method is shown in Figure \ref{fig:kapproval}. This is similar to the standard approval voting method, but with a cap on the number of alternatives a voter can choose.
\begin{definition}[K-approval voting]
\begin{itemize}
 \item Each voter chooses (``approves") at most $K$ projects. 
 \item The projects are ranked in descending order according to the total number of approving voters.
 \item The outcome is decided by picking projects in this order until the budget is exhausted.
 \end{itemize}
  \end{definition}
  
   While this method is simple and has wide acceptance, it doesn't make the voters factor in the costs of the projects into their votes. Consider the following example:
\begin{example}\label{ex:kapproval}
Think of one voter (or a homogeneous pool of voters) doing $1$-approval. Let's say the budget is \$300 and there are three projects on the ballot: Project $A$, Project $B$, and Project $C$. Their costs are \$ $300$, \$ $200$, \$ $100$, respectively. Imagine that Project $A$ is the single most beneficial project among all, but Project $B$ and Project $C$ are together more valuable in terms of benefit than just Project $A$. 

\begin{table}[ht]
	\centering
	\caption{An example, 3 projects and additive utilities.}
    \label{tab:kapprovalexample}
    \begin{tabular}{ | r | r | r | r |}
    \hline
      & Cost \\ \hline
    Project $A$ &  \$$300$ \\ \hline
    Project $B$ & \$$200$ \\ \hline
    Project $C$ & \$$100$\\ \hline
    \end{tabular}
\end{table} 
It is natural to expect that the voter(s) will vote for Project $A$ under $1$-approval. Although Projects B and C are together feasible under the budget of \$ $300$, there is no way for the voter to express these preferences.
Although this is a very simple example which makes a strong assumption on the voters' response to the K-approval rule, it is an apt illustration of observed voting behavior (see Section \ref{subsec:knapsackdata}). 

A point to note here is that even truthful reporting under K-approval voting (choosing projects with highest utility) does not lead to good outcomes as seen in Example \ref{ex:kapproval} above. It is conceivable that the outcome might be different if the voters vote strategically, taking aggregation with the budget constraint into account. However, we would like a natural way for the voters to express their preferences with respect to costs and benefits -- one that works at least for a homogeneous population under truthful reporting.

This leads us to the following questions:
\begin{itemize}
\item What is a good way of eliciting voter preferences?
\item What utility models can we use to uncover some insights on the above problem?
\end{itemize}
\end{example}

We utilize ideas derived from ways of solving the Knapsack Problem: corresponding to our Knapsack Voting and Value-for-money schemes, which will be treated in detail in later sections of this paper Sections \ref{sec:knapsack} and \ref{sec:vfmcomps}. The remainder of this section is organized as follows: we introduce some modeling preliminaries in Section \ref{subsec:votutilmodels} and give a short overview of our work in the Section \ref{subsec:overview}. We discuss the implications and properties of our model in Section \ref{subsec:discutil}

\subsubsection{Other implementations similar to Knapsack Voting}\label{subsec:otherindep}
We would like to mention here that independent of our work on Knapsack Voting, other implementations of the same voting method have been in use in some places in Europe -- this adds to the relevance of our work in studying the advantages of Knapsack Voting. The most prominent (with a total allocated budget, across many districts, of roughly \EUR{40} million) among these are the ``shopping cart votes" use in the PB process in Madrid\footnote{\url{https://decide.madrid.es/presupuestos-participativos-resultados}}, Spain, where it has been in use since 2016. The origins of this method come from the ``Open Active Voting" method in the My Neighbourhoods Project in Reykjavik\footnote{\url{https://citizens.is/portfolio_page/my-neighbourhood/}}, Iceland, where it has been in use since 2012. The same platform has been used in other places like K\`opavogur and Mos\`o in Iceland\footnote{\url{https://citizens.is/portfolio/}}, and a pilot implementation in Argyll and Bute\footnote{\url{https://www.argyll-bute.gov.uk/abpb/} ; \url{https://citizens.is/portfolio/}}, Scotland.

\subsection{Modeling assumptions: voting and utility models}\label{subsec:votutilmodels}
The participatory budgeting problem addresses the following scenario: the residents of a city, collectively the set of voters $\V$, vote on a set $\p$ of projects that they have identified to be worthwhile, where project $j \in \p$ has a cost $c_j$ and there is a fixed total budget of $B$ Dollars.

\subsubsection{Voting models}
We will mainly consider a fractional model of voting -- each voter can allocate the budget $B$ in any which way among the projects. Our results for this model also carry over (approximately) to an integral setting. For ease of exposition, we defer a discussion of the above to Section \ref{app:integral} in the Appendix.

\begin{definition}[Fractional Vote]\label{def:fracvote}
Each voter $v$ chooses an allocation $\{w^v_p\}_{p \in \p}$ such that $0 \leq w^v_p \leq c_p$ for all $p \in \p$, and $\sum_{p \in \p} w^v_p = B$.
\end{definition}

In effect, we make the following assumption:
\begin{assumption}\label{ass:lastproject}
Every project is allowed to be fractionally implemented, and the budget is completely used up.
 \end{assumption}
 
\emph{Note:} It may or may not be realistic to partially fund projects (e.g., renovate only one floor of a library rather than the entire building), but partial funding and the requirement of spending all of the budget are both common in practice.

\subsubsection{Utility models}
In what follows, we give concise definitions of the utility models we use in this paper. A full discussion of their properties is deferred to Section \ref{subsec:discutil}.

Let $\{w_p\}_{p \in \p}$ denote an outcome, and $\{w^v_p\}_{p \in \p}$ denote voter $v$'s ideal allocation. We first define the $\ell_1$ costs model.
\begin{definition}[$\ell_1$ costs]\label{def:l1model}
The $\ell_1$ cost of an outcome $\{w_p\}_{p \in \p}$ (where $\sum_{p \in \p} w_p = B$: Assumption \ref{ass:lastproject}) for a voter $v$ is given by $\sum_p |w_p - w_p^v|$.
\end{definition}

We define another model called ``Overlap utility".
\begin{definition}[Overlap utility]\label{def:overlap1}
The Overlap utility of an outcome $\{w_p\}_{p \in \p}$ (where $\sum_{p \in \p} w_p = B$) for a voter $v$ is given by $\sum_{p \in \p}\min \{w_p,w_p^v\}$.
\end{definition}
The above model will be especially useful in our proofs in Section \ref{sec:knapsack}. The term corresponding to $p$ in the defintion above is equal to the ``overlap" between $w^v_p$ and $w^*_p$, and so we call this expression the Overlap Utility. This means, in other words, that the utility from $p$ for a voter $v$ is equal to $w_p^v$ when $w_p^v \le w_p^*$, and $w_p^*$ otherwise. For example, if $w^v_p = 5$ and $w^*_p = 10$, then the utility derived by $v$ is $5$. 

We also define an additive concave utility model, which is a generalization of the Overlap utility model.
\begin{definition}[Additive concave utility]\label{def:addconcave}
In the fractional vote model, the additive concave utility of an outcome $\{w_p\}_{p \in \p}$ (where $\sum_{p \in \p} w_p = B$) for a voter $v$ is given by $\sum_{p \in \p} f_v^p(w_p)$, where each $f_v^p(.)$ is a non-decreasing concave function.
\end{definition}

 We use the definitions from above to summarize our results in the next section. A discussion of the properties of the above utility models is deferred to Section \ref{subsec:discutil}.

\subsection{Our contributions}\label{subsec:overview}
Our overall goal in this paper will be to present voting mechanisms that are particularly suited to the Participatory Budgeting problem, especially the Knapsack Voting scheme. These mechanisms, by getting voters to either optimize over the total budget, or compare projects based on their value-for-money, enable the voter to inherently consider costs and benefits, and moreover, are implementable using interactive digital tools. 
In building our case for Knapsack Voting, we aim to establish its advantages as follows:
\begin{enumerate}
\item Knapsack Voting, rather than $K$-approval voting, is the correct application of the approval voting principle to the participatory budgeting problem, and thereby avoids some issues of strategic voting.
\item It is easy to implement in a user-friendly way, and we have been able to deploy it in real elections successfully.
\item Empirically it agrees more with binary comparisons of value-for-money: a natural proxy of voter preferences.
\item It nudges voters to better consider cost-benefit trade-offs by making them recognize the budget constraint.
\end{enumerate}

\begin{figure}
\centering
  \includegraphics[width=240pt]{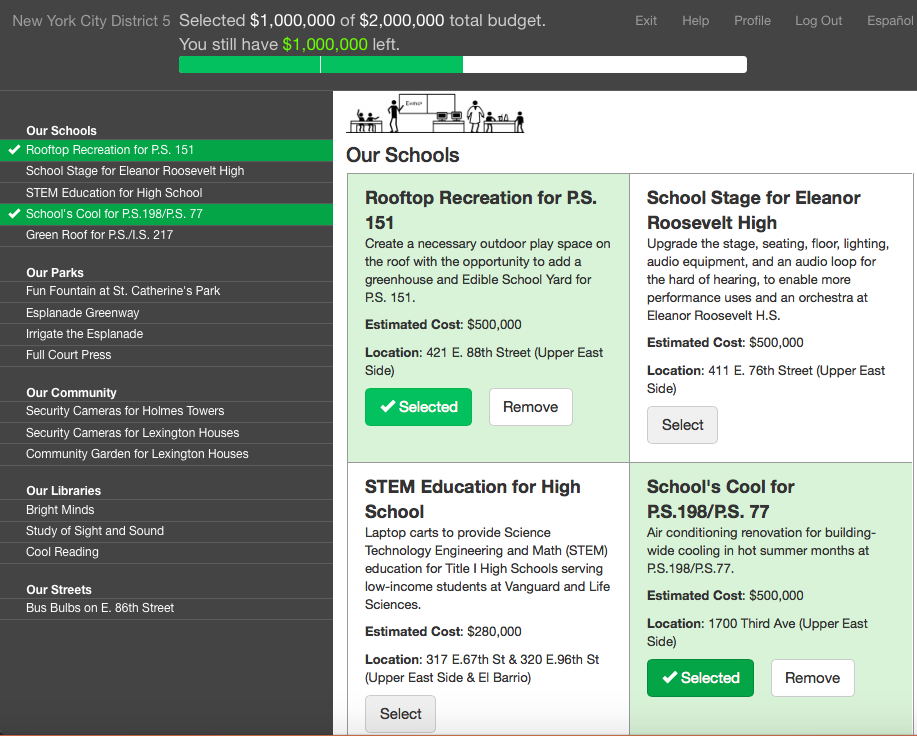}
  \caption{Our interface for Knapsack vote: example from PB NYC District 5 2015}
  \label{fig:knapsack}
\end{figure}

\subsubsection{Knapsack Voting}
Our primary scheme is called Knapsack Voting. The main idea behind this scheme is that each voter has to adhere to the budget constraint, thus keeping costs in mind while giving her vote -- thus internalizing the global constraints, while giving their preferences toward the outcome. Our interface for this method is shown in Figure \ref{fig:knapsack} (integral) and Figure \ref{fig:fracknapsack} (fractional).

We primarily look at the Knapsack Vote mechanism in the fractional setting. Most of our results extend to the integral case in an approximate fashion -- we defer a discussion of this setting to Section \ref{app:integral} in the Appendix.

We now give an informal definition of Knapsack Voting (for a formal one see Definition \ref{def:fracknapsack}).
\begin{definition}[\sc{Knapsack Vote}] \label{def:fracknapinformal}
The Knapsack Vote mechanism is defined as follows:
\begin{itemize}
\item  Each voter $v \in \V$ votes for an allocation $\{w^v_p\}_{p \in \p}$, such that $\sum_{p \in \p} w^v_p = B$.
\item The outcome is chosen according to approval scores in a per-dollar sense, as in Figure \ref{fig:fracvote} (Section \ref{subsec:discutil}).
\end{itemize}
\end{definition}
As we can see, while the elicitation method is very different, the aggregation method here is similar to K-approval in that it uses approval scores.

We analyze the strategic properties of Knapsack Voting assuming an $\ell_1$ cost model, in which, loosely speaking, the cost to a voter depends on how much the outcome differs from her preferred allocation as determined by the $\ell_1$ distance between the two (Definition \ref{def:l1model}). Under this model (and also the related Overlap utility model, via the equivalence Lemma \ref{lem:ell1overlap}), we show that Knapsack Voting is strategy-proof and welfare-maximizing (Theorems \ref{thm:EMDstratproof}, \ref{thm:welfmax}). 

\begin{result}\label{res:stratproof}
Under the $\ell_1$ cost model (and also the Overlap utility model), Knapsack Voting is strategy-proof and welfare-maximizing.
\end{result}

This result is analogous to the strategy-proofness results for Approval Voting under dichotomous preferences \citep{brams2007approval}. 

We also extend Knapsack Voting to more general settings with revenues, surpluses or deficits (as opposed to one with a fixed budget) where approval voting has no known analog, and prove similar results (Section \ref{subsec:rsd}, Theorem \ref{thm:budgbal}). 

Also, under the more general case of additive concave utilities, we characterize a weaker yet interesting notion of strategy-proofness of a voter's best response, based on the concept of sincerity \citep{niemi1984problem}: it is in a voter's best interest to vote for projects that she favors among those that are winning without her vote (Theorem \ref{thm:parttrue}). 

\begin{result}\label{res:partstratproof}
Under additive concave utilities, a voter's best response under Knapsack Voting is partially strategy-proof.
\end{result}

Together, these amount to substantial evidence that Knapsack Voting aligns the incentives of the voters with that of the decision maker. Note that these two results do not hold for K-approval voting\footnote{The utility functions are defined with respect to the voter's favorite budget allocation and the budget outcome chosen, and not what the voter expresses. Under such models, the results do not hold for K-approval -- e.g. it is not strategyproof under such utility models.}.

An interesting, and standard, way of understanding voting rules is viewing them as Maximum Likelihood estimators (MLEs) \citep{conitzer2012common,young1988condorcet}.  Votes are assumed to be drawn from a suitable noisy model parametrized by a ``ground truth" outcome, such that the voting rule is the MLE of the ``ground truth" given any realization of votes. For example, the Kemeny Young rule is the MLE of rankings drawn from the Mallows model. The Mallows model defines a probability distribution over all possible rankings, with the probability of each ranking depending on how much it differs from a given  ground truth order.  Since we are concerned with the problem of determining an outcome which is a subset of projects that satisfies the budget constraint, we will define a model that takes such a set as the ground truth, and defines a probability for all valid sets, depending on how much they differ from the ground truth. By interpreting Knapsack Voting (Section \ref{sec:knapsackml}) as the MLE of this natural noise model, we reinforce the aggregation method of Knapsack Voting.

\subsubsection{Value-for-money comparisons} In our second scheme, we elicit the voters' preferences based on their perceived value-for-money from projects. By value-for-money we mean the utility of a project normalized by its cost: if $v_{i,p}$ is the utility of voter $i$ from project $p$, then its value-for-money as perceived by $i$ is $\frac{v_{i,p}}{c_p}$. See Figures \ref{fig:bftb},\ref{fig:paperballot} for voting interfaces based on this idea, where voters compare pairs of projects, or rank their top projects, based on value-for-money. 

In the previous section, we discussed how most elections use paper ballots for those voters that opt out of the digital platform. Using the idea of value-for-money, we designed a paper ballot that accompanies Knapsack Voting. 

We will discuss this and other details about value-for-money in Section \ref{sec:vfmcomps}.

\begin{figure}
    \centering
    \includegraphics[width=200pt]{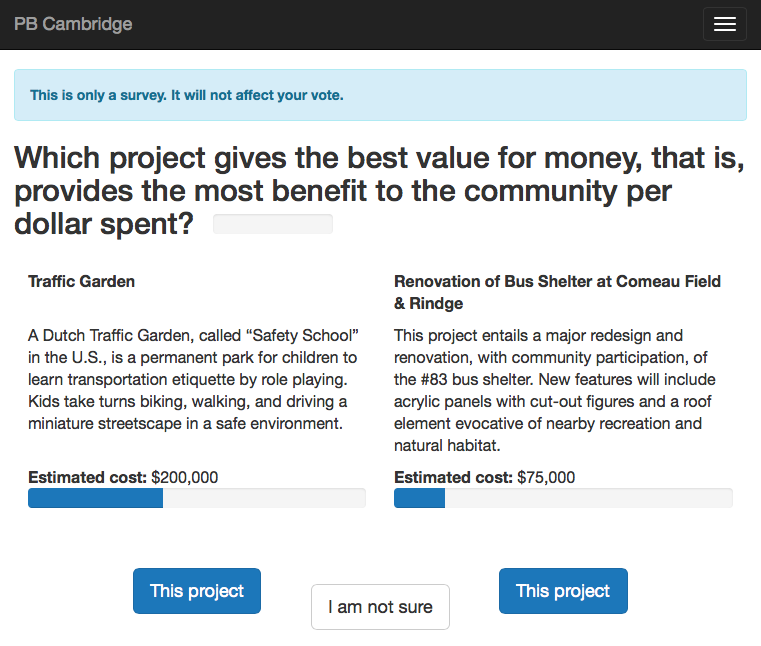}
    \caption{Value-for-money comparison: example from PB Cambridge 2015}
    \label{fig:bftb}    
\end{figure}

\subsubsection{Deployments and Data Analysis}\label{subsubsec:dataanalysis}
Using our digital voting system that has gained widespread acceptance in many cities/municipalities  across the nation, we tested our methods across various participatory budgeting elections. In most of these elections, K-approval was used as the official ballot, and in addition, we tested either the Knapsack or the value-for-money comparisons (Figure \ref{fig:bftb}) or both. We do value-for-money to estimate the aggregate pairwise preferences of voters.

We describe our experimental procedure in greater detail in Section \ref{sec:experiments}. Based on the data from our experiments, we make the following observations:
\begin{enumerate}
\item Knapsack Voting leads to a more economical consideration of the projects as compared to K-approval (\ref{subsec:knapsackdata}), 
\item The time taken by a voter to do Knapsack Voting is similar to that for K-approval. Value-for-money comparisons take much less time (Section \ref{subsec:timing}).
\item Knapsack Voting does better than K-approval in terms of agreement with  pairwise comparisons (Section \ref{subsec:knapvskapp}).
\end{enumerate}
%\item Kemeny-Young method is a practical way of aggregating the Value-for-money comparisons into a full ranking on the projects (Empirical Hypothesis \ref{hyp:kemenylp} in Section \ref{subsec:kemeny}).

We present empirical evidence in support of these observations in Section \ref{sec:experiments}. 

The above observations suggest a significant qualitative difference between the outcomes of Knapsack Voting and K-approval. In particular, Observation 3 above suggests that the outcome of Knapsack Voting is more in alignment with the aggregate preferences of the voters than that of K-approval. 

Observation 1 just reinforces the fact that Knapsack Voting leads to a more economical consideration of projects by the voters. It is natural to expect voters to pay more attention to costs of projects under Knapsack Voting.
And Observation 2 suggests that Knapsack Voting does not involve a much larger cognitive load on the voters than K-approval.  

However, we stop short of claiming that Knapsack Voting leads to outcomes that are more beneficial to society as a whole. It could be worthwhile to think of ways to compare the outcomes in terms of their long-term societal value.

\subsection{A discussion of utility models}\label{subsec:discutil}
As mentioned before, one of our goals will be to prove some good properties of Knapsack Voting, in particular strategy-proofness. To do so, we first need to make some modeling assumptions.\footnote{With no specific assumptions, strategy proof mechanisms are impossible, see Section \ref{subsec:related}.}

\begin{figure}
  \centering
  \includegraphics[width=230pt]{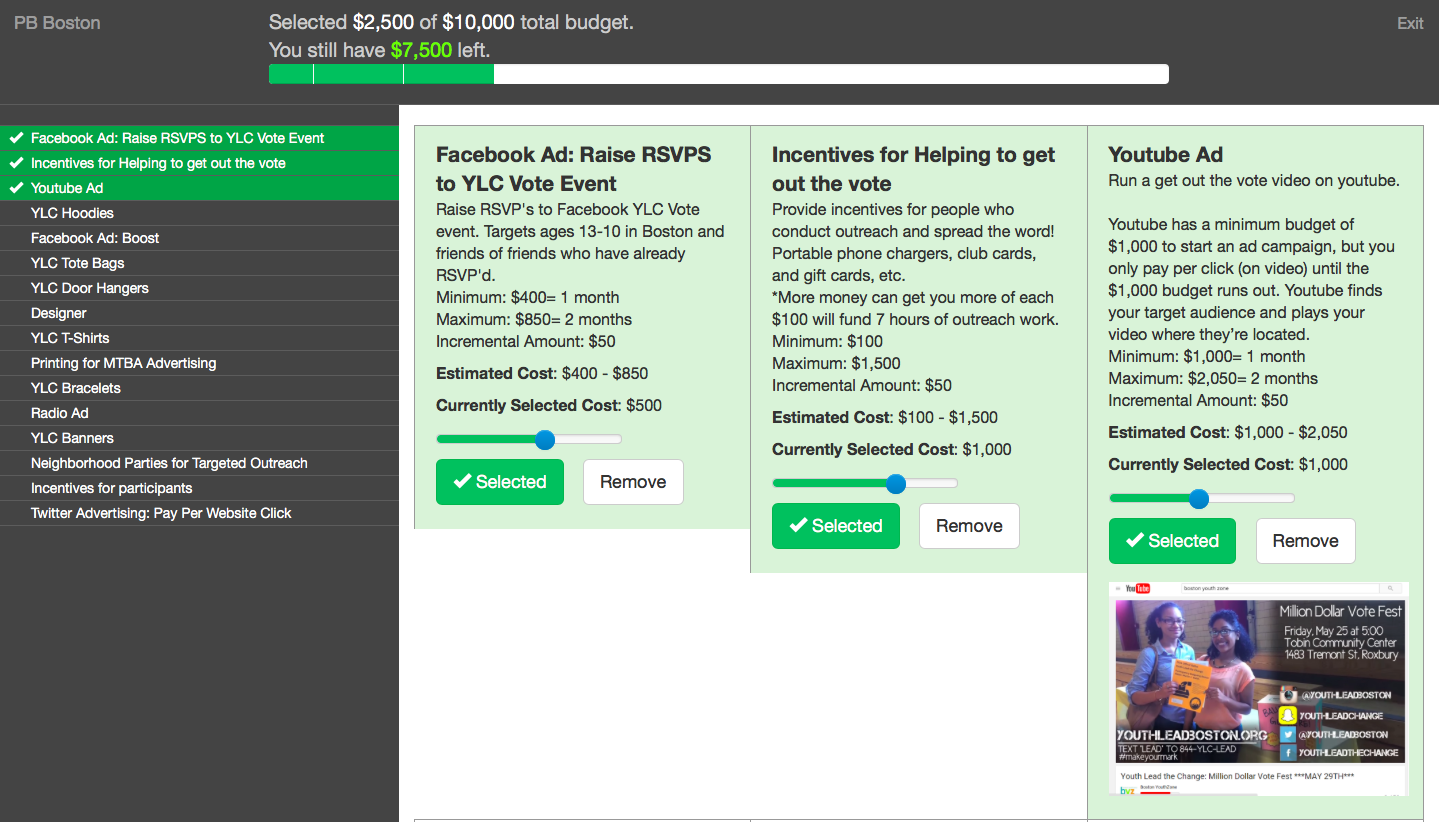}
  \caption{Experimental interface for Fractional Knapsack}
  \label{fig:fracknapsack}
\end{figure}

There is a wide range of budget problems where the allocation per project is not a fixed cost, but can be variable, possibly up to a certain upper limit (see Figure \ref{fig:fracknapsack} for an experimental interface). This setting is also very relevant to PB\footnote{see \url{https://pbstanford.org/boston16internal/knapsack} for a PB election where most of the projects allowed for flexible allocations}. The fractional setting is also relevant outside of PB. For example, the federal budget consists of a variable allocation of money to various sectors of the economy at large such as education, healthcare, defense, etc. To accommodate variable allocation, and to avoid combinatorial difficulty, we consider a fully fractional version of Knapsack Voting using a ``per-dollar" approach.
\begin{figure}[h!]
  \begin{minipage}{0.4\textwidth}
    \centering
    \includegraphics[scale=0.3]{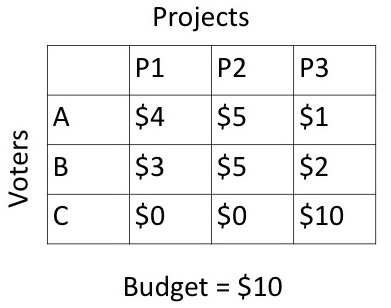}
    \caption{Preferred allocations of voters}
    \label{fig:frac_table}
  \end{minipage}
  \hspace{20pt}
  \begin{minipage}{0.5\textwidth}
  \centering
    \includegraphics[scale = 0.25]{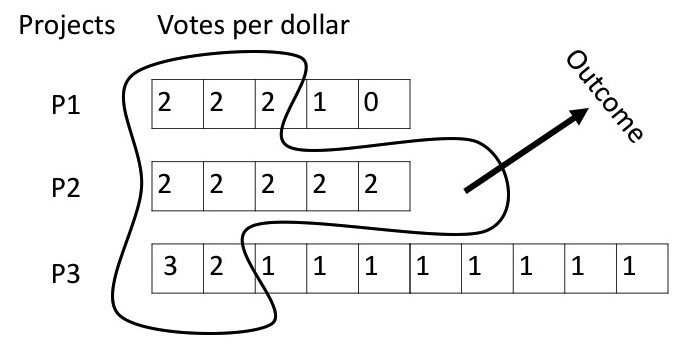} 
    \caption{Votes per dollar for each project}
    \label{fig:fracvote}
  \end{minipage}
\end{figure}
\begin{example} \label{eg:perdollar}
Given a budget of $10$, and three voters $A,B,C$, let's say that the preferred allocations of the voters for projects $P1,P2,P3$ of maximum costs $(5,5,10)$ are $(4,5,1)$, $(3,5,2)$, $(0,0,10)$ respectively (Figure \ref{fig:frac_table}). We divide each project into as many dollars as its maximum cost, and given that $A$ allocates $4$ to $P1$, we give one vote to each of the first $4$ dollars of $P1$, and so on. Figure \ref{fig:frac_table} shows how many votes each dollar of each project gets. If we select those dollars with the highest number of votes, we get the allocation $(3,5,2)$. Notice that here we didn't have to break ties to get a feasible allocation.
\end{example}

As mentioned earlier, we study the above problem under different utility models. We now discuss them below -- these models were formally defined earlier in Section \ref{subsec:votutilmodels}.

\subsubsection{$\ell_1$ costs and Overlap utility}
In Section \ref{sec:knapsack}, we look at a model where the disutility of a voter for a given allocation of the budget is given by how much her preferred allocation differs from it in terms of the $\ell_1$ distance between the two allocations. Minimizing $\ell_1$ distances to reach a compromise decision has some precedence in public policy literature \citep{andre2010designing}. Assume that the ideal preference of a voter $v$ corresponds to allocating an (integral) amount $w_p^v$ to project $p \in \p$, such that $\sum_{p \in \p} w^v_p = B$, where $B$ is the total budget. Then we know from Definition \ref{def:l1model} that the cost for a voter $v$ for an outcome $\{w_p\}_{p \in \p}$ is given by $\sum_{p \in \p} |w^v_p - w_p|$.

We must note here that applying the $\ell_1$ cost model to the case where outcomes are allowed to be of size smaller than $B$ (i.e. violating Assumption \ref{ass:lastproject}, e.g., the purely integral setting) does not lead to anything meaningful. Indeed doing so could violate the free disposal of utilities.\footnote{We thank one of the reviewers for pointing this out via the described example.}
\begin{example}\label{eg:freedisp}
Consider a total budget of $2$, with three projects $a$, $b$ and $c$. Take a voter that has a preferred allocation of $1$ to $a$. Then an allocation to $b$ and $c$ of $1$ each has a cost of $-3$ for the voter, whereas just an allocation of $1$ to project $b$ has a cost of $-2$. 
\end{example}

As mentioned in Definition \ref{def:overlap1}, we define another model called Overlap Utility, which we will see is the analog of dichotomous preferences as applied to approval voting \citep{brams2007approval}: for a voter's preferred allocation, we look at each dollar subproject (see Section \ref{sec:knapsack} for a formal definition of subprojects) in it, and count it as $1$ towards the utility if the same dollar subproject is in the outcome allocation and $0$ if not.

However, it is easy to see that the Overlap utility model satisfies\footnote{Even in the approximately integral model mentioned in Appendix \ref{app:integral}}
\begin{enumerate}
\item convexity of preferences, and
\item free disposal of utilities.
\end{enumerate}

To see why it satisfies free disposal, let's look at Example \ref{eg:freedisp} again, where a voter has a preferred allocation of 1 to $a$. An allocation to $b$ and $c$ of 1 has a utility of 0 since there is no overlap. If the voter agreed completely with this allocation, she would have a utility of 2. And the allocation differs from her ideal point with respect to allocating 2 units of the budget, thereby giving a net utility of 0. If we look at an allocation of 1 to just $b$, then the utility is still 0. If the voter agreed completely with this allocation, she would have a utility of 1. And the allocation differs from her ideal point with respect to 1 unit of the budget, again giving her a net utility of 0.

In the fractional setting, the $\ell_1$ model is equivalent to the Overlap utility. This equivalence (see Lemma \ref{lem:ell1overlap}) is useful in the proof of Result \ref{res:stratproof}. 

\subsubsection{Additive concave utilities}
The most general model we consider is one where voters have utilities that are additive over the projects in $\p$.  Moreover, the utility from each project is concave in the amount of allocation given to that project. It is easy to see that this model (see Definition \ref{def:addconcave}) is a generalization of the Overlap utility model. The Overlap utility is obtained as a special case as follows: given a voter's ideal allocation $\{w_p^v\}_{p \in \p}$, $f_p^v(w_p)$ is defined to be $\min \{w_p,w_p^v\}$. We use this model in Section \ref{subsec:parttrue} to prove our partial strategy-proofness result (Result \ref{res:partstratproof} in the previous section). 

\subsection{Related Work}\label{subsec:related}
A wide variety of voting procedures, both ranked and non-ranked, have been studied in social choice literature \citep{brams2002voting}. Plurality and Approval Voting among non-ranked procedures, and Borda, Copeland and Kemeny-Young \citep{levin1995introduction} among ranked, are perhaps the most well known. In approval voting, each voter can vote for, or ``approve of", any number of candidates on the ballot. Each candidate gets a  vote from every voter that voted for her, and the candidate with the most votes wins. Approval voting has been widely studied as an alternative to plurality voting, with some compelling advantages \citep{brams1993approval}. Various modifications to approval voting, such as added constraints \citep{brams1990constrained}, and cumulative voting \citep{bhagat1984cumulative}, that help improve the representation of minorities or under-represented groups in committees have been examined. In this paper, we devise the right adaptation of the approval paradigm to participatory budgeting elections.

Another parallel stream of literature in social choice theory has delved deeply into the question of manipulability of voting rules. It has been  shown that the only reasonable social choice functions with single/multiple winners that are not susceptible to strategic manipulation by voters are degenerate forms of dictatorships \citep{gibbard1973manipulation,satterthwaite1975strategy}. Such impossibility results \citep{duggan2000strategic} rule out the existence of strategy-proof mechanisms for our setting in general. 

However,  strategy-proof schemes were shown to be possible in restricted domains like single-peaked preferences \citep{moulin1980strategy,barbera1993generalized,nehring2002strategy}. Inspired by Black's Median Voter Theorem \citep{black1948decisions,moulin1980strategy} characterized Generalized Median Voter Schemes as the only strategy-proof rules under single-dimensional single-peaked preferences. This was generalized to multiple dimensions for preferences on the cartesian box \citep{border1983straightforward,barbera1993generalized}, and general subsets of the cartesian box \citep{barbera1997voting}. The results of \citep{barbera1997voting} imply that, under single-peaked preferences, the only voting rules that could be be both strategy-proof and non-dictatorial for our problem are Generalized Median Voter Schemes (GMVS). Unfortunately, GMVSs do not respect the budget constraint for three or more projects.

We propose an aggregation rule that finds the geometric median of the votes on the simplex that satisfies the budget constraint and show that strategy-proofness holds under specific utility models like the $\ell_1$ cost model (Definition \ref{def:l1model}), an adaptation of dichotomous preferences to our setting. It is also interesting to compare our setting with the axioms of Arrow's Impossibility Theorem \citep{arrow2012social}, especially the Independence of Irrelevant Alternatives (IIA). With a weaker form of IIA \citep{campbell2000weak}, Knapsack Voting, akin to Approval Voting with dichotomous preferences, satisfies all three of Arrow's axioms. In recent work\footnote{which appeared while our manuscript was under review.}, Freeman et al. characterize a broad class of strategy-proof ``moving phantom" mechanisms under the $\ell_1$ cost model \citep{freeman2019truthful}. Interestingly, their social-welfare maximizing method (which is equivalent to Knapsack Voting upto tie-breaking) is the unique Pareto-optimal mechanism in the above-mentioned class.

From the more recent literature in Computational Social Choice, our work is similar in spirit to some recent work on identifying the trade-offs between the utilities from various possible societal activities \citep{conitzer2015crowdsourcing}. Our work also involves capturing trade-offs, but only so far as to obtain the best allocation of the budget, and we propose voting methods that use ideas from the Knapsack Problem to do so. There has also been some work on selecting committees under weight or cost constraints given complete rankings from voters \citep{klamler2012committee}. However, our work is different in that we do not get complete rankings from voters, and design ways of eliciting preferences of voters. Under the utilitarian distortion framework, there has been some recent work on designing a mechanisms for PB \citep{benade2017preference,danitruthful}. However, the outcomes here could be a factor $\sqrt{m}\log m$ off of the optimal, and in general, do not extend to the fractional case. Another drawback of such mechanisms is that they are not transparent enough to be used in formal ballots. 

There has also been some work on developing visualization tools for voters to make sense of different items in complex budget problems, e.g. the federal budget \citep{kim2016budgetmap}. In our setting, the items on the ballot are simple and well defined, for instance building a park, or improving streets, and can be described concisely on our digital voting platform for voters to make an informed decision.

The problem of choosing the best subset of candidates given global preference information has received some attention lately \citep{klamler2012committee,lu2011budgeted}, and various methods have been proposed. For the problem of selecting optimal sets of weighted candidates given a sum weight constraint, there is some work on algorithms that, given the knowledge of a global preference relation among the candidates, compute optimal sets for various objective functions \citep{klamler2012committee}. Another treatment of this problem involves a recommendation that pairs each voter to one of these candidates, thereby assuming that each voter derives satisfaction from one of the candidates in the chosen subset \citep{lu2011budgeted}. In our context, eliciting full rankings on the candidates is not feasible and, moreover, we motivate our methods in terms of social welfare, with each voter benefiting from all the chosen candidates. There is also some work \citep{fain2016core} with an overlapping set of co-authors, that studies the fairness properties of mechanisms for participatory budgeting, largely using
the notion of core.

\section{Knapsack Voting: Imposing budget constraints}\label{sec:knapsack}
 In this section, we show that Knapsack Voting is strategy-proof (Theorem \ref{thm:EMDstratproof}) and welfare-maximizing (Theorem \ref{thm:welfmax}) under the $\ell_1$ costs model. By Lemma \ref{lem:ell1overlap}, the same results hold for the Overlap utility model too. We then extend Knapsack Voting to more complex settings with revenues, deficits or surpluses (Equation \ref{eqn:knapsackbudgbal}) for which there is no known analog of Approval voting. We prove that Knapsack Voting is strategy-proof in this setting (Theorem \ref{thm:budgbal}). We also prove a partial strategy-proofness result in the case of additive concave utilities (Theorem \ref{thm:parttrue}).

 In what follows, we adopt the Fractional vote model, as in Definition \ref{def:fracvote}. More formally, since each project is fractionally implementable, we split each project $p \in \p$ into $c_p$ different ``per-dollar sub-projects"  $D^p_{1}, D^p_{2}, \ldots, D^p_{c_p}$, and collect all the sub-projects into a set $\pp$. The problem then reduces to choosing $B$ out of the $C$ candidates in $\pp$. Given a vote from a voter $v$ using our Knapsack interface (Figure \ref{fig:knapsack}), it is guaranteed to be \emph{consistent} after the per-dollar conversion, i.e., a set $S_v \subset \pp$ such that for all $p \in \p$, if $D^p_t \in S_v$, then $D^p_{t^\prime} \in S_v$ for all $1 \leq t^\prime < t$. 

 For any $j \in \pp$, define its score as $\score(j) \triangleq |\{ v \in \V : j \in S_v \}|$. We will use a \emph{consistent deterministic tie-breaking order}: a strict ordering $\prec$ on $\pp$ such that for all $p \in P$, $1 \leq t^\prime < t \leq c_p \iff D^p_{t^\prime} \prec D^p_t$ (if $j \prec k$, then $j$ gets priority over $k$ when $\score(j) = \score(k)$).
%For use in later sections, let us denote by $\sim$, a relation on $\pp$ which defines which candidates correspond to the same project in $\p$.
\begin{definition}[Knapsack Vote] \label{def:fracknapsack}
The Knapsack Vote mechanism is defined as follows:
\begin{itemize}
    \item Each voter $v \in \V$ submits a consistent subset $S_v \subseteq \pp$, such that it satisfies a budget constraint  $|S_v| = B$.
    \item The winning set is given by $\underset{S : |S| = B}{\arg \max} \sum_{j \in S} \score(j)$, using a consistent deterministic tie-breaking order $\prec$.
\end{itemize}
\end{definition}
Note that given consistent votes, and a consistent tie-breaking order, the Knapsack Vote gives a ``consistent" outcome. 

We also redefine Overlap utility (Definition \ref{def:overlap1}) in terms of the per-dollar notation above -- this will be of use shortly.
\paragraph{Overlap utility}
Let $S_v \in \pp$ and $S^* \in \pp$ be the allocations in the per-dollar sense corresponding to $\{w^v_p\}_{p \in \p}$ and $\{w^*_p\}_{p \in \p}$.  Restating Definition \ref{def:overlap1} in terms of subprojects:
\begin{definition}\label{def:overlap}
The utility of the voter $v$ is given by $|S_v \cap S^*|$, or equivalently, $\sum_{p \in \p} \min\{w^v_p,w^*_p\}$. 
\end{definition}
In this quantity, the term corresponding to $p$ is equal to the ``overlap" between $w^v_p$ and $w^*_p$, and so we call this expression the Overlap Utility.

\subsection{Strategy-proofness under the $\ell_1$ cost model}\label{subsec:stratproof}
Given the set of projects $\p$, let's say that the true preference $S_v$ of each voter $v$ corresponds to allocating an (integral) amount $w^v_p$ to project $p \in \p$, such that $\sum_{p \in \p} w_p = B$. Any outcome that deviates from this allocation would result in some dis-utility for the voter. Let $\{w^*_p\}_{p \in \p}$ denote the final outcome. By Definition \ref{def:l1model}, the dis-utility that each voter gets is equal to the $\ell_1$ distance between $\{w^v_p\}_{p \in \p}$ and $\{w^*_p\}_{p \in \p}$, i.e., $\sum_{p \in \p} |w^v_p - w^*_p|$.

%\begin{definition}[$\ell_1$ cost model]\label{def:l1model}
%Given an outcome $\{w^*_p\}_{p \in \p}$, the utility of a voter $v$ with a preferred allocation $\{w^v_p\}_{p \in \p}$ is given by $U(v) = - \sum_{p \in \p} |w^v_p - w^*_p|$.
%\end{definition}
%how much overlap there is between the outcome, and her preferred budget allocation. More formally, the utility of the voter $v$ from any project $p$ is given by $\min\{w^v_p, w^*_p\}$, and consequently her total utility is $\sum_{p \in \p} \min \{ w^v_p,w^*_p\}$. We define this equivalently in the per-dollar sense as the following:

In the above, the budget allocation to each project $p$ is flexible, and can take any integral value in $[0,c_p]$, with the votes are restricted to the simplex determined by $\sum_{p \in \p}w_p = B$. The outcome of Knapsack Voting can be defined as an $\ell_1$-median restricted to the above-mentioned simplex (follows from Lemma \ref{lem:ell1overlap} and Theorem \ref{thm:welfmax}), i.e., an allocation that minimizes the sum of the $\ell_1$ distances to the votes. \footnote{This works in continuous space too, with a suitably defined tie-breaking rule} Under the Overlap utility model, by defining the outcome as the geometric $\ell_1$ median, we get the following result:

\begin{theorem} \label{thm:EMDstratproof}
For Knapsack Voting, under the $\ell_1$ cost model, voting for her preferred allocation $S_v$ is a weakly dominant strategy for voter $v$.
\end{theorem}

\begin{proof}
Let $S_{-v}$ and $\score_{-v}(\cdot)$ be the outcome determined by the votes of everyone except $v$ using the Knapsack Voting rule (Definition \ref{def:fracknapsack}). As mentioned earlier in this section, for any $i \in \pp$, $\score_{-v}(i) = |\{u \in \V \setminus \{v\} : i \in S_u\}|$, where $S_u$ denotes the vote of any voter $u \in \V \setminus \{v\}$ (since we are looking for a dominant strategy for voter $v$, $S_u$ here can be any valid vote, not necessarily the truthful report of voter $u$).  Assume that $T_v \neq S_v$ is a best response of $i$. Let the outcome after incorporating $T_v$ be $\score(\cdot)$ and $S$. 

Let $j \in T_v \setminus S_v$ such that if $j = D^p_t$ for some $p \in \p$ (recall the definition of per-dollar sub-projects from earlier in this section), then $D^p_{t^\prime} \notin T_v$ for all $t^\prime > t$.  Choose some $k \in S_v \setminus T_v$ such that if $k = D^q_z$ for some $q \in \p$, then $D^q_{z^\prime} \in T_v$ for all $z^\prime < z$. Such a $k$ exists because $T_v$ and $S_v$ are consistent and of the same size $B$. Let $T_v^\prime \triangleq T_v \cup \{k\} \setminus \{j\}$, and the outcome here be $\score^\prime(\cdot)$ and $S^\prime$. We will show that $T_v^\prime$ is also a best response for $v$. If $S=S^\prime$, then we have nothing to prove. In what follows, we will assume $S \neq S^\prime$.

We have that $\score^\prime(j) = \score(j) - 1$, $\score^\prime(k) = \score(k) + 1$, and for all $l \in \pp \setminus \{j,k\}$, $\score^\prime(l) = \score(l)$. Note that the only change here is that the score of $j$ decreases, and the score of $k$ increases. As a result, for any given consistent tie-breaking order (as discussed at the head of this section), $S^\prime \setminus S$ must be singleton. Further, we must have either $j \in S \setminus S^\prime$ or $k \in S^\prime \setminus S$, i.e., any change in outcome must involve either $j$ moving from within the winning set to without, or $k$ from without to within (or both).  

If $j \in S \setminus S^\prime$, it must be that $S \setminus \{j\} \subset S^\prime$, since for any $l \in S \setminus \{j\}$, $\score^\prime(l) \geq \score(l)$. And since $S^\prime \setminus S$ is a singleton set, say it contains $m \in \pp$ (possibly $m=k$), the change in utility of voter $v$ is $\indicator(m \in S_v) - \indicator(j \in S_v) = \indicator(m \in S_v) \geq 0$ (here $\indicator(\cdot)$ is the indicator function that takes the value $1$ if the argument is true, and $0$ if not). In other words, the change in utility from removing $j$ is 0, and that from adding any $m$ in its place cannot be negative.

Similarly, if $k \in S^\prime \setminus S$, then $S \setminus S^\prime = \{m^\prime\}$ for some $m^\prime \in \pp$. And the change in utility is $\indicator(k \in S_v) - \indicator(m^\prime \in S_v) = 1 - \indicator(m^\prime \in S_v) \geq 0$.

By repeating this process until we only have elements in $S_v$, we have not decreased the utility of voter $v$. Therefore, the utility obtained by voting for $S_v$ cannot be strictly dominated by that for any $T_v \neq S_v$.
\end{proof}

By a simple modification, the above proof extends to the Overlap utility model as well.
\begin{corollary}\label{cor:EMDstratproofoverlap}
For Knapsack Voting, under the Overlap utility model, voting for her preferred allocation $S_v$ is a weakly dominant strategy for voter $v$.
\end{corollary}

%Definition \ref{def:knaputil} represents a situation where the voters have sharp preferences on the possible budget outcomes. The above theorem states that in such a situation, the Knapsack Rule admits a dominant strategy.

Under truthful voting, the Knapsack Voting outcome maximizes the social welfare. 
\begin{theorem}\label{thm:welfmax}
The truthful dominant strategy equilibrium for Knapsack Voting is welfare-maximizing under $\ell_1$ costs.
\end{theorem}

Before proving the above theorem, we will show an equivalence between the $\ell_1$ cost model and the Overlap utility model.

\subsubsection{Overlap utility}
Let $S_v$ and $S^*$ be the allocations in the per-dollar sense corresponding to $\{w^v_p\}_{p \in \p}$ and $\{w^*_p\}_{p \in \p}$. In the budgeted case being treated in this section (cf. Section \ref{subsec:rsd}), the utility given by the $\ell_1$ model is equal to the Overlap utility. This equivalence is of independent interest, and is also useful in the proof of Theorem \ref{thm:EMDstratproof}. We now prove this equivalence formally.

\begin{lemma}\label{lem:ell1overlap}
$|S_v \cap S^*| = B - \frac{1}{2} \sum_{p \in \p} |w^v_p - w^*_p|$.
\end{lemma}
\begin{proof}
As discussed above, we know that $|S_v \cap S^*|= \sum_{p \in \p} \min \{ w^v_p,w^*_p\}$. Partition $\p$ into two sets $\p_L$ and $\p_H$, where $\p_L \triangleq \{p \in \p: w^v_p > w^*_p\}$ is the set of projects for which the amount allocated in the outcome is less than that in the voter $v$'s preferred allocation, and $\p_H \triangleq \{p \in \p: w^v_p \leq w^*_p\}$ is the set of projects where the amount allocated in the outcome is at least as much as that in the voter $v$'s preferred allocation. Then we have $|S_v \cap S^*| = \sum_{p \in \p_L}w^*_p + \sum_{p \in \p_H}w^v_p$.
 First we note that $\sum_{p \in \p_L}w^*_p + \sum_{p \in \p_H} w^v_p = \sum_{p \in \p}w^*_p - \sum_{p \in \p_H}(w^*_p - w^v_p) = B - \sum_{p \in \p_H}(w^*_p - w^v_p)$. Similarly, $\sum_{p \in \p_L}w^*_p + \sum_{p \in \p_H}w^v_p =  - \sum_{p \in \p_L}(w^v_p - w^*_p) + \sum_{p \in \p}w^v_p = - \sum_{p \in \p_L}(w^v_p - w^*_p) + B$.
 Therefore $|S_v \cap S^*| = \frac{1}{2}(2B - \sum_{p \in \p_L}(w^v_p - w^*_p) - \sum_{p \in \p_H}(w^*_p - w^v_p) = B - \frac{1}{2} \sum_{p \in \p} |w^v_p - w^*_p|$.
\end{proof}
%\begin{corollary}\label{cor:intvoteequiv}
%Under Assumption 1, for the Integral Vote model, we have 
%$|S_v \cap S^*| = \frac{B + B_v}{2} - \frac{1}{2} \sum_{p \in \p} |w^v_p - w^*_p|$, where $B_v = \sum_{p \in S_v} c_p$.
%\end{corollary}

\begin{proof}[of Theorem \ref{thm:welfmax}]
The Knapsack Voting rule outputs a set $S^*$ that maximizes $\sum_{j \in S} \score(j)$ among all consistent sets $S$ that satisfy the budget constraint.

If voters vote truthfully, then 
\begin{align*}
\sum_{j \in S^*} \score(j) = \sum_{j \in S^*} \sum_{ v \in V }\indicator (j \in S_v) = \sum_{v \in V} \sum_{j \in S^*} \indicator (j \in S_v) = \sum_{v \in V} |S_v \cap S^*|,
\end{align*}
and by Lemma \ref{lem:ell1overlap}, the proof follows.
\end{proof}

In effect, by way of Lemma \ref{lem:ell1overlap}, the equivalence of $\ell_1$ costs and Overlap utilities model, we have established the following:
\begin{corollary}\label{cor:intvotestratoverlap}
Knapsack Voting is strategy-proof and welfare-maximizing under either $\ell_1$ costs or Overlap utilities.
\end{corollary}

% We refer to the term $\frac{1}{|V|}\sum_{v \in \V} |S_v \cap S^*|$ as the \emph{average dollars per voter}. We can use this measure to motivate a statistical maximum likelihood interpretation of Knapsack Voting. We do this in the Appendix (\ref{app:knapsackml}) to not interrupt the flow in this section. While Theorem \ref{thm:EMDstratproof} explicitly assumes the Overlap utility model based on this quantity, it can be used as a measure of social welfare to compare the outcomes of K-approval and  Knapsack Voting in real data. This motivates the following hypothesis:
%\begin{hypothesis}\label{hyp:dollarspervoter}
%The average dollars per voter is higher under Knapsack Voting than K-approval.
%\end{hypothesis} 

We would like take a moment here to discuss Knapsack Voting in terms of the conditions in Arrow's Theorem. Choosing a set of projects that fit the budget can be thought of as a partial ordering $\prec$ between those within, and those without. In our case, if the preferences of voters between two projects, say $a$ and $b$ with respect to this partial ordering remains the same, then the final score of $a$ and $b$ remains the same. If the outcome is $a \prec b$, then changing the other preferences cannot result in $b$ being chosen over $a$, i.e., change from $a \prec b$ to $b \prec a$. A similar property has been characterized as a weaker form of the Independence of Irrelevant Alternatives in the context of Approval Voting \citep{campbell2000weak}.

 We must also mention here that under the $\ell_1$ cost model, Knapsack Voting is \textbf{not group strategy-proof}. For example, consider the following:
\begin{example}
There are 5 projects $\{a,b,c,d,e\}$, 4 voters $\{1,2,3,4\}$, and a budget of \$2. The true preferences are: Voters 1 and 2 allocate \$2 to $a$, and 3 allocates \$1 to $b$ and \$1 to $c$, and 4 allocates \$1 to $d$ and \$1 to $e$. 
If they voted truthfully, the outcome is \$2 to $a$, giving a utility of $0$ to both 3 and 4. Assuming ties are broken in the order $b,d,c,e,a$. Then if voters 3 and 4 both voted for \$1 to $b$, and \$1 to $d$, then the outcome changes to $b$ and $d$, giving both voters a utility of 1.
\end{example}

\subsection{An extension to scenarios with revenues, deficits or surpluses}\label{subsec:rsd}
A similar result can be obtained in a case where there is no hard budget, and there are both expenditure terms, and revenue terms by extending the framework of approval voting under dichotomous preferences. Such considerations are common  in real settings that are a little beyond PB, e.g. designing mechanisms for determining the federal budget. There has been some work as a follow-up to ours, \citep{garg2017collaborative} with an overlapping set of co-authors, which uses models with revenue and spending as categories for the design and testing of new adaptive voting mechanisms. 

For our purposes, a voter proposes both how to generate revenue from among various avenues in $\mathcal{R}$, and how to spend it on various projects in $\p$. We will now discuss a case where the budget is balanced. Extensions to cases where there are multiple revenue items, or where the budget is unbalanced, are easy to see. 

As before, in the per-dollar sense, let $\pp$ be the set of per-dollar sub-projects for expenditure, and $\rr$ be the set of per-dollar sub-projects for revenue. We use a similar $\ell_1$ cost metric for the project terms as for the budgeted case, with revenue included. Formally, we will assume that if $\{x^*_r\}_{r \in \mathcal{R}}$ is the outcome revenue level, and the voter $v$ prefers $\{x^v_r\}_{r \in \mathcal{R}}$, then her dis-utility from the revenue term will be $\sum_{r \in \mathcal{R}} |x^*_r - x^v_r|$, and her total dis-utility will be $\sum_{p \in \p} |w^v_p - w^*_p| + \sum_{r \in \mathcal{R}} |x^*_r - x^v_r|$. Because we are considering \emph{balanced budgets}, we will assume $\sum_{r \in \mathcal{R}} x^v_r = \sum_{p \in \mathcal{P}} w^v_p$ for all voters $v$.

Let $R_v \subseteq \rr$ and $S_v \subseteq \pp$ denote the vote of voter $v$ satisfying $|R_v| = |S_v|$. For $j \in \rr$ we define $\score(j) \triangleq - |\{v \in \V: j \notin R_v\}|$. The Knapsack Vote outcome $R^* \subseteq \rr$ and $S^* \subseteq \pp$ is defined as follows:
\begin{align}\label{eqn:knapsackbudgbal}
(R^*,S^*) = {\arg\max}_{(R,S) : |S| = |R|} \left( \sum_{i \in S} \score(i) + \sum_{j \in R} \score(j) \right)
\end{align}

  In this setting, revenue corresponds to voters paying the government in fees/taxes, and hence we make the score corresponding to revenue terms negative. We will need a consistent tie-breaking rule here as well. But since the size of the outcome is not fixed, we can just break ties in favor of the bigger (or smaller) sets. We state a similar result as in Theorem \ref{thm:EMDstratproof} for this case:
 
 \begin{theorem}\label{thm:budgbal}
 With a balanced budget, Knapsack Voting is strategy-proof under the $\ell_1$ cost model. 
 \end{theorem}
 We defer the proof to the appendix (\ref{proof:budgbal}), as it is similar to the proof of Theorem \ref{thm:EMDstratproof}, and not crucial to the exposition here.

We can also extend this to \emph{unbalanced} settings by including the budget deficit (expenditure - revenue) as part of the voters' preferences. If $\Delta_v = |S_v| - |R_v|$ is the preferred deficit of voter $v$, and the deficit in the outcome is $\Delta^*$, then the voter incurs an additional dis-utility of $| \Delta^* - \Delta_v|$, i.e., how much the deficit in the outcome differs from her preferred level. As such, it is equivalent to adding an additional revenue term.

 \subsection{``Partial" strategy-proofness under additive concave utilities}\label{subsec:parttrue}
In the case with additive concave utilities, we illustrate a weaker, yet interesting property which we call ``partial" strategy-proofness. Consider a focal voter $i$ responding to the votes of all others. Assume that she has full knowledge about the how the others voted in aggregate. If $S_{-i}$ denotes the cumulative votes of all voters except $i$, she knows  
$\W_{-i} \subseteq \pp$, the winning set as determined by $S_{-i}$. Let
$\W(S_i,S_{-i})$ denote the set of winners if her vote $S_i$ is added. A \emph{best response} for $i$ is a consistent set
\begin{displaymath}
S_i^{\star} \in \underset{S_i \subseteq \pp : |S_i|=B}{\arg \max} \sum_{j \in \W(S_i,S_{-i})} \vp_{i,j},
\end{displaymath}

where $\vp_{i,j}$ is the utility of voter $i$ for sub-project $j \in \p$. As mentioned in Definition \ref{def:addconcave}, the utility from per-dollar sub-projects is monotone non-increasing (as the utility from a project is concave in the amount allocated), i.e., if $x<y$ and $a=D^p_x$ \footnote{$D^p_x$ is the $x$-th dollar subproject in project $p$.} and $b=D^p_y$ for some $p \in \p$, then $\vp_{i,a} \geq \vp_{i,b}$. With respect to voter $i$, we say a candidate $k$ \emph{dominates} $j$ if and only if 
\begin{displaymath}
 k \in \W_{-i} \mbox{ and } \vp_{i,k} \geq \vp_{i,j}.
\end{displaymath}

A consequence of the budget constraint is that it allows for \emph{partial strategy-proofness} in the best response of a focal voter responding to all other votes. 
\begin{theorem}[Partial Strategy-proofness]\label{thm:parttrue}
 Under Knapsack voting, there exists a best response $S_i^\star$ such that if $k$ \emph{dominates} $j$, and $j \in S_i^\star$, then $k \in S_i^\star$. 
\end{theorem}
The proof is not critical to the exposition here, and we defer it to the Appendix.
%Given a best response $S_i^\star$, if 
%$p \in S_i^\star$, %\cap \W_{-i}^c \cap \W(S_i^\star,S_{-i})$, 
%then there is another best response $S_i^{\star\star}$ such that 
%%$ \{p\}\cup
%$(\Sigma_{i,p} \cap \W_{-i}) \subseteq S_i^{\star\star}$
%where 
%$\Sigma_{i,p} \triangleq \{j \in \p : \frac{v_{i,j}}{c_j} > \frac{v_{i,p}}{c_{p}} \}$.

We can think of $\W_{-i}$ as representing the candidates that $i$ thinks are popular. And projects with a higher ``value-for-money" are preferred from her perspective. The theorem states that a simple way for a voter to act in her best interest is to vote for projects that to her are both \textbf{popular} and \textbf{favorite}. This notion is similar in spirit to the ideal of \emph{sincerity} in approval voting \citep{brams2007approval,niemi1984problem}. 

The reason partial strategy-proofness holds under {Knapsack Voting} is that the voters face the same constraints as the outcome they are collectively deciding. It is interesting to note that this property does not hold under {K-approval} voting, as there is a mismatch between the constraints on the voters (choosing a fixed number of projects $k$) and those on the outcome (a fixed budget $B$). We illustrate this with a small example.

\begin{example}\label{eg:partstratapp}
Consider $5$ projects $a,b,c,d,e$ of costs \$200, \$100, \$100, \$100, \$200, and using the 2-approval rule. The budget is \$400 and they poll $100,50,50,50,20$ votes respectively without counting in $i$'s vote. The tie is broken in the order $e,d,c,b,a$, and in this case $d$ wins over $b,c$.
Let's say that $i$'s utilities for these are $500,100,150,200,500$ respectively. Based on the above, $i$'s best response is to not vote for $a$ but for $c,d$. 
\end{example}

%In addition to analyzing the strategic properties of this scheme, we would like to characterize empirically how the \emph{knapsack voting} method affects voting behavior as compared to \emph{K-approval} voting. It is natural to expect that if voters voted in term of the total utility of each project without considering its cost, as in K-approval, projects of larger costs are likely to be over-represented; as opposed to Knapsack Voting which leads to a more economical consideration of the projects. This leads us to the following: 
%\begin{hypothesis}\label{hyp:avgcost}
%The outcome of Knapsack Voting is less biased towards projects of larger cost compared to that of K-approval.
%\end{hypothesis} 

\subsection{Maximum Likelihood Interpretation}\label{sec:knapsackml}
One way of looking at voting rules is as follows: There exists a ``ground truth" outcome, and each voter has a noisy perception of it. And the the voting rule is the Maximum likelihood estimator of the ``ground truth" given any realization of votes. The {Knapsack} voting rule can be interpreted in this way since it belongs to a family of voting rules known as scoring rules \citep{conitzer2012common}. Selecting subsets of winners based on the MLE approach has received some attention lately \citep{procaccia2012maximum} wherein voting rules to select subsets based on their performance on various metrics with respect to noisy comparisons or rankings drawn from a Mallows model.

 In this section, we will use the per-dollar approach and explicitly construct a natural noise model for the votes for which it is the Maximum Likelihood estimator. The noise model we construct is similar to the Mallows model, but defines a distribution over subsets directly as opposed to rankings.

\begin{definition}[Noisy Knapsack Vote Model]
There is ``ground truth" set $S^\star \in \pp$, which satisfies $ |S^\star| = B $. Each vote $S_i$ is drawn i.i.d. according to a distribution that is given by:
\begin{eqnarray*}
   &\Pr \{ S_i | S^\star\}  \propto \exp ( |S^\star \cap S_i|), & \mbox{ if } |S_i| \leq B \\
    &\Pr \{ S_i | S^\star\} = 0,               & \text{otherwise}
\end{eqnarray*}
\end{definition}

The quantity $|S^\star \cap S_i|$ is equal to the number of dollars in the allocation given by $S^\star$ that agrees to that given by voter $i$ in her vote $S_i$, and by Lemma \ref{lem:ell1overlap} is related to the Overlap/$\ell_1$ utility we discussed previously.

By taking the logarithm of the probabilities $\Pr(S_i|S^\star)$, it is easy to see that the Maximum Likelihood estimate of the ``ground truth" $S^\star$ given all the votes, is the set $S \subseteq \pp$ satisfying $|S| = B$ that maximizes  
\begin{equation}\label{eqn:knapsackml}
\frac{1}{|V|}\sum_{i \in \V} |S \cap S_i|.
\end{equation}

\begin{theorem}
The \emph{Knapsack rule} returns the maximum likelihood estimate under the Noisy Knapsack Vote Model.
\end{theorem}
\begin{proof}
Let's rewrite the quantity in Equation \ref{eqn:knapsackml} as 
\begin{align*}
&\sum_{i \in \V} |S \cap S_i|  = \sum_{i \in \V} \sum_{j \in S} \indicator(j \in S_i)
= \sum_{j \in S}\sum_{i \in \V}  \indicator(j \in S_i)  = \sum_{j \in S} \score(j),
\end{align*}
where $\indicator(.)$ is a $\{0,1\}$ variable that takes on a value $1$ when the statement in the argument is true, and $0$ otherwise. This quantity is maximized by picking $B$ candidates from $\pp$ that have the highest score (see Definition \ref{def:fracknapsack}), which is essentially what the Knapsack rule chooses.\qed
\end{proof}

\section{Voting based on \emph{value-for-money}}\label{sec:vfmcomps}
We have seen that Knapsack Voting  has many advantages with respect to its strategic properties, and implementation in Participatory Budgeting elections. We now turn to another way of eliciting voters' preferences in this setting - \emph{value-for-money} comparisons. Given a single agent's Knapsack Problem, one way of computing the optimal solution is to order the items according to their value-to-size ratio and pick the higher ranked ones in order till the knapsack capacity is used up. This order can be ascertained by comparing pairs of projects according to their value-to-size ratio. Asking voters to compare/rank projects based on their \emph{value-for-money} is a natural analog (see Figure \ref{fig:bftb}) of this idea in a setting where multiple agents together have to decide an outcome.

Unlike Knapsack Voting, aggregation schemes based on value-for-money comparisons cannot be guaranteed to have good strategic properties. Further, these schemes are not as transparent as Knapsack Voting, since different information is elicited from each voter, and the aggregation method is not as straightforward as that of Knapsack Voting. Value-for-money schemes also do not extend naturally to settings with revenue, deficits and surpluses. 

Despite these difficulties, they are useful in practice because they can be used to design paper ballots with Knapsack Voting, and also elicit the aggregate preferences of voters to make empirical observations from data. We will discuss these presently.

In addition, Value-for-money schemes have the following potential advantages:
\begin{itemize}
\item a smaller cognitive load on voters: especially with large ballots (see Section \ref{subsec:timing})
\item aggregation in cases where the budget is not known or fixed a priori (perhaps using Kemeny-Young like ranking rules)
\end{itemize}
We will delve into these in more detail in the next section.

\begin{figure}[ht]
\centering
  \includegraphics[width=170pt]{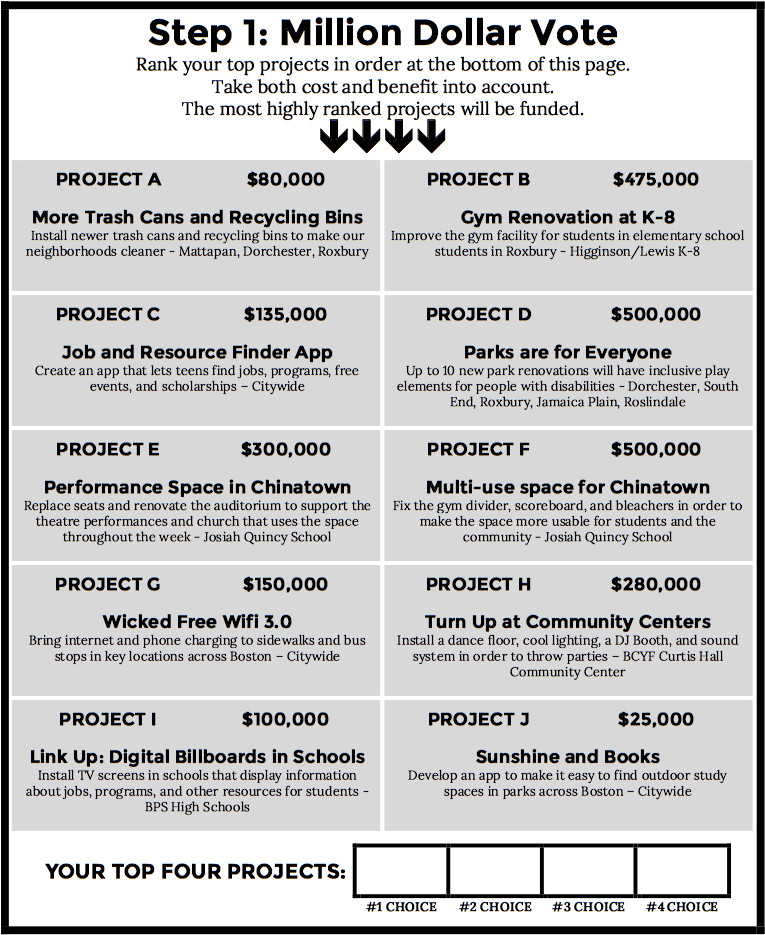}
  \caption{Value-for-money ranking paper ballot for Knapsack Voting - PB Boston 2016}
  \label{fig:paperballot}
\end{figure}
%\subsection{Value-for-money}
%The first idea we will look at is asking each voter to rank the projects on the ballot in decreasing order of their value-for-money. Since ranking all the projects, providing a complete ordering would be a tedious process, we will elicit partial rankings, say of size $k$, from the voters, i.e., each voter lists her top $k$ projects according to value-for-money. The aggregation of partial rankings has attracted a lot of attention over the years \citep{fagin2006comparing,gonzalez2001aggregation}. For our purposes, we will aggregate these votes by giving a single point to each of the $k$ projects that the 
%In the Youth Lead The Change 2016 PB election held in Boston, a set of projects was put to vote to youth between the ages 12 and 25. As mentioned earlier, we implemented Knapsack Voting as the official ballot process in this election. The voting was not done completely digitally and the city wanted to provide paper ballots as well. Since it is tedious to do Knapsack Voting on paper, we designed a paper ballot that asked voters to rank their top 4 projects according to value-for-money (Figure \ref{fig:paperballot}). From a total of around 4000 voters, about 10\% used the digital Knapsack interface, and the rest used the paper ballot with ranking.
\subsection{Value-for-money rankings and paper ballots}\label{subsec:vfmrankings}
As mentioned before, some elections that employ digital voting also require a corresponding paper ballot. This was the case in the Youth Lead The Change 2016 PB election held in Boston, a set of projects was put to vote to youth between the ages 12 and 25. We implemented Knapsack Voting as the official ballot process in this election. Since it is tedious to do Knapsack Voting on paper, we designed a paper ballot that asked voters to rank their top 4 projects according to value-for-money (Figure \ref{fig:paperballot}). From a total of around 4000 voters, about 10\% used the digital Knapsack interface, and the rest used the paper ballot with ranking. 

Here we ask each voter to rank the projects on the ballot taking both costs and benefits into account \footnote{There is an implicit assumption here that taking both costs and benefits into account is the same as considering value-for-money}. Voters take into account costs and benefits into account, and ranking reinforces this consideration. 
%Since ranking all the projects, providing a complete ordering would be a tedious process, we will elicit partial rankings, say of size $k$, from the voters, i.e., each voter lists her top $k$ projects according to value-for-money. The aggregation of partial rankings has attracted a lot of attention over the years \citep{fagin2006comparing,gonzalez2001aggregation}. For our purposes, we will aggregate these votes by giving a single point to each of the $k$ projects that the voters rank.
Interpreting these votes as Knapsack votes, we observe that the outcome from the rankings was exactly the same as that of Knapsack Voting. This gives us a way of using rankings to solve the budgeting problem in practice.

\subsection{Value-for-money comparisons}\label{subsec:vfmcomps}
We can elicit fine-grained information about the preferences of voters between pairs of projects by doing value-for-money comparisons between randomly chosen projects.

\begin{definition}[Value-for-money comparison]  \label{def:vfmcomps}
For each pair of projects $\{j,k\}$ from $\p$ shown to her, voter $i$ chooses a winner 
$w_i(\{j,k\}) = {\arg \max}_{t \in \{j,k\}}\; \frac{v_{i,t}}{c_t}$.
\end{definition}

Doing these comparisons in practice, leads us to an interesting empirical observation of the structure in the aggregate preferences of voters. We observe that compiling this data from our experiments reveals a (nearly) transitive majority relation among the projects (Figure \ref{fig:cambridge}). By this we mean that there exists a Condorcet winner, and upon removal of that candidate, there again exists a Condorcet winner, and so on. 

\begin{figure}[ht]
\begin{minipage}{0.5\textwidth}
\centering
  \includegraphics[width=200pt]{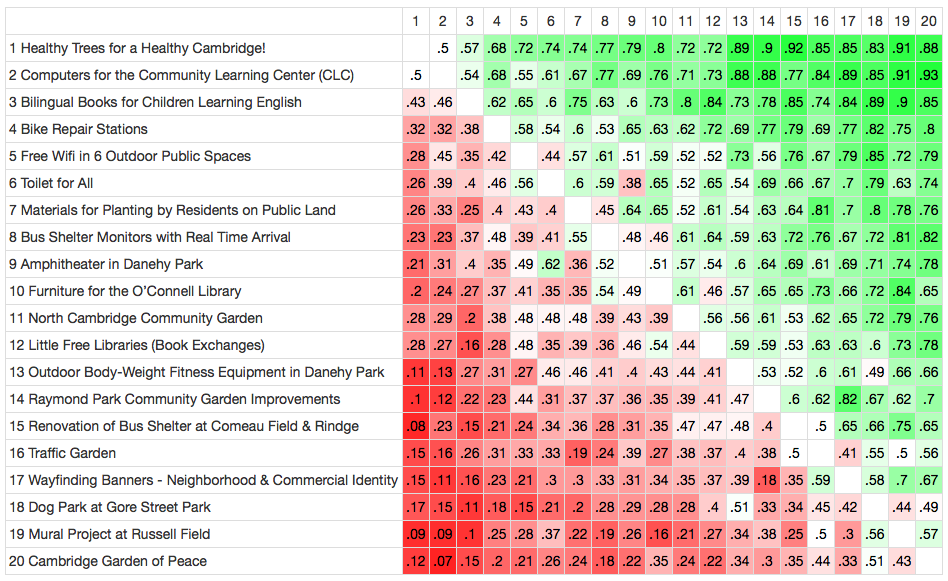}
  \caption{Cambridge 2014 - Comparison Table}
  \label{fig:cambridge}
\end{minipage}%
\begin{minipage}{0.5\textwidth}
\centering
  \includegraphics[width=200pt]{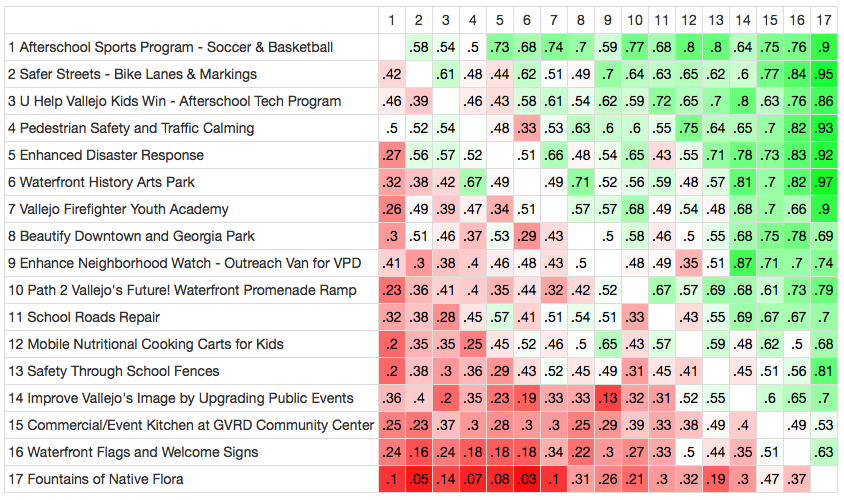}
  \caption{Vallejo 2015 - Comparison Table}
  \label{fig:vallejo}
\end{minipage}%
\end{figure}

In figures \ref{fig:cambridge} and \ref{fig:vallejo}, we show the aggregate strength of each comparison across the two elections -- Cambridge 2015 and Vallejo 2015. More details about the two elections are in the next section. The number in the cell denotes the fraction of comparisons where the row project beats the column project. Green represents a fraction greater than 0.5, and red represents a fraction smaller than 0.5. The darker the green(red), the closer it is to 1(0).

Given this structure, there is a clear order among projects based on the aggregate strength of the comparisons. In fact, any Condorcet rule leads to the same outcomes on such data. If this structure does not hold, then the aggregation of these comparisons becomes more complicated.

\subsubsection{Comparing outcomes of Knapsack and K-approval}\label{subsec:bordaset}
We will define a quantitative measure for any outcome that represents the level of its agreement with the pairwise comparisons.To do this, we generalize the Borda rule and define a Borda score for outcomes that are sets of projects.
 
In our experiments, we compare one project against another as shown in Figure \ref{fig:bftb}, and to each voter we present randomly chosen pairs so as to ascertain the aggregate preferences of voters. Given the results of these comparisons, let $n(j,k)$ denote the number of voters that chose $j$ over $k$.
%The above defintion is equivalent to Definition \ref{def:vfmcomps} in a per-dollar sense. In practice we don't compare candidates such that $j \sim k$ (i.e.,``dollars" corresponding to the same project). 

We now build on the standard for a standard definition of Borda's rule \citep{young1974axiomatization}, and its classical interpretation as an MLE \citep{young1988condorcet}. Generalizing the Borda score for a single winner, we define
\begin{definition}[Set-Borda score] \label{def:setborda}
For any set of projects $S \subseteq \p$, given $C = \sum_{p \in S} c_p$ and $M = \sum_{p \in \p}c_p$, the  Set-Borda score of $S$ is given by $\frac{1}{C(M - C)} \sum_{j \in S} \sum_{k \notin S} c_j c_k \left( n(j,k) - n(k,j) \right)$
\end{definition}

We can interpret a pairwise value-for-money comparison between two projects $j$ and $k$, as a pairwise relation between a dollar sub-project of $j$ and a dollar sub-project of $k$. The Set Borda score of a set of $S$ corresponds to the average number of such dollar versus dollar comparisons that agree, minus those that disagree, with the partition induced by it. We can use this score, as a measure of social welfare that embodies the essence of the Borda rule, to empirically compare the outcome of Knapsack Voting with that of K-approval. 
 
 %Note that these comparisons do not use the knowledge of the budget constraint, and we can therefore use the comparisons to define a full ranking of the projects. A classical ranking rule that comes mind is the Kemeny-Young, and in Appendix \ref{subsec:kemeny} we show how we can adapt it to this setting. In fact, we see that the LP relaxation of the Kemeny-Young problem results in integer optimal values on data from our experiments, thereby giving us the exact Kemeny-Young order. 

%\begin{hypothesis}\label{hyp:borda}
%The outcome of Knapsack Voting has a higher Set Bor
%\end{hypothesis} 
% This is perhaps due to the inherent similarity in accounting for costs for the projects in determining the outcome. This adds evidence to the suitability of Knapsack Voting to participatory budgeting elections.

\section{Our digital platform, and results from our experiments}\label{sec:experiments}
Our work with Participatory Budgeting began with a partnership with Chicago's 49th Ward to develop a digital voting system for their election. Since then we've have worked with almost a dozen cities/districts over the past couple of years (see \url{http://pbstanford.org}), but we will primarily look at data collected from the following elections:
\begin{itemize}
\item Boston 2016 - Youth Lead the Change
\item Boston 2015 - Youth Lead the Change
\item Cambridge 2015
\item Cambridge 2014
\item Vallejo 2015
\item New York City District 5 2015 (NYC5)
\item New York City District 8 2015 (NYC8)
\end{itemize} 
The total number of voters in each of these elections were 4176, 2600, 3273, 2194, 1834, 704, 271 respectively.

In the Boston 2016 election, we implemented Knapsack Vote as the official ballot mechanism. In the other elections, besides implementing a interactive interface for the official K-approval vote (Fig. \ref{fig:kapproval}), our goal has been to experiment with Knapsack voting and Value-for-money comparisons and use the data so collected to complement a theoretical understanding of these methods. After the K-approval vote used for the formal election, we ask the voters to participate in our experiments on either Knapsack Voting or Value-for-money comparisons or both. 
\begin{itemize}
\item In Cambridge 2015 and Boston 2015, we showed each voter either Knapsack and value-for-money experiments with a 50\% chance. 
\item In NYC5 and NYC8, we did only the Knapsack vote.
\item In Cambridge 2014 and Vallejo 2015 we did only pairwise comparisons.
\end{itemize}
 The Knapsack interface (Figure \ref{fig:knapsack} in Section \ref{subsec:overview}) has a live budget bar that shows how much of the budget has used up with the current selection. For the value-for-money comparisons, we present each voter with a fixed number of pairs of projects chosen uniformly at random without replacement, and ask them the following question (in the spirit of Definition \ref{def:vfmcomps}): ``Which of these projects gives a higher benefit to the community per dollar spent?" (see Figure \ref{fig:bftb}). We must mention that only a subset of the entire set of voters actually chose to take part in our experiments, and we report this percentage in Tables \ref{tab:time_knapsack} and \ref{tab:time_vfm}. This procedure was approved by Stanford University's Institutional Review Board.

\begin{figure}[ht!]
\begin{minipage}{0.47\textwidth}
\centering
  \includegraphics[width=180pt]{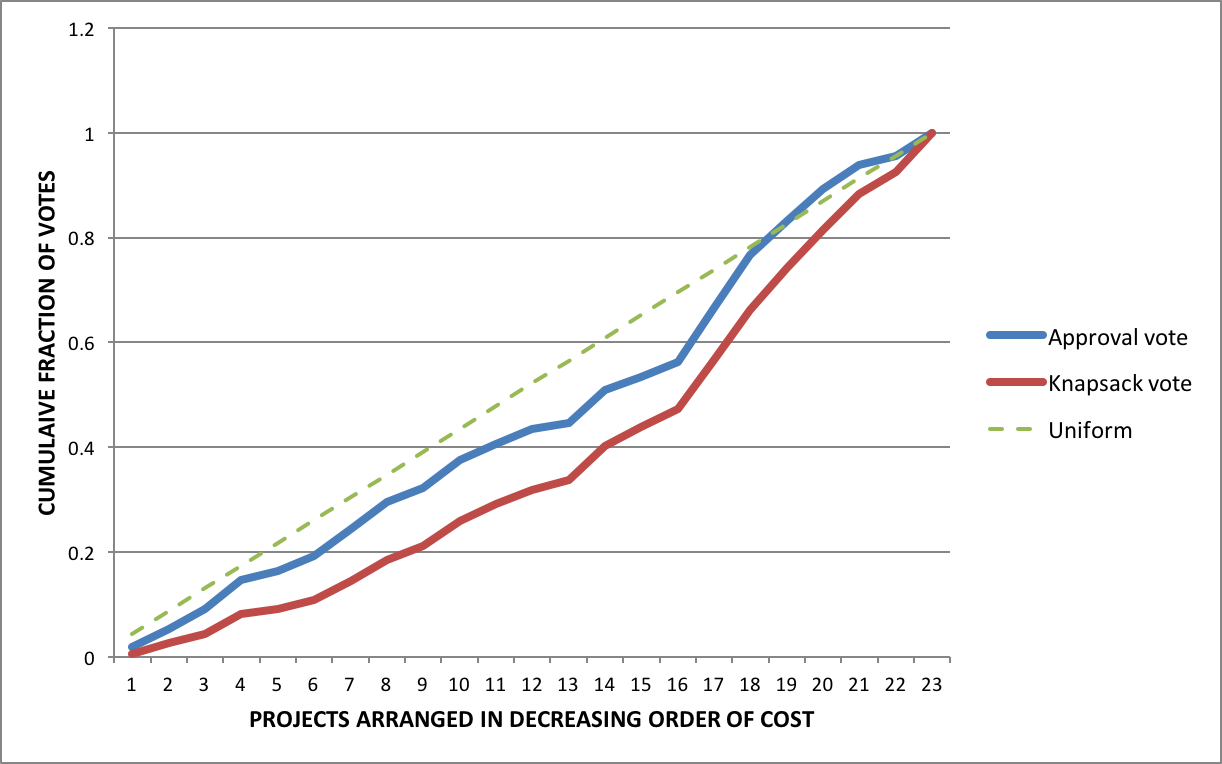}
  \caption{Cumulative fraction of total votes versus cost of projects: Cambridge 2015}
  \label{fig:cambridge_cumulative}
\end{minipage}%
\begin{minipage}{0.47\textwidth}
\centering
  \includegraphics[width=180pt]{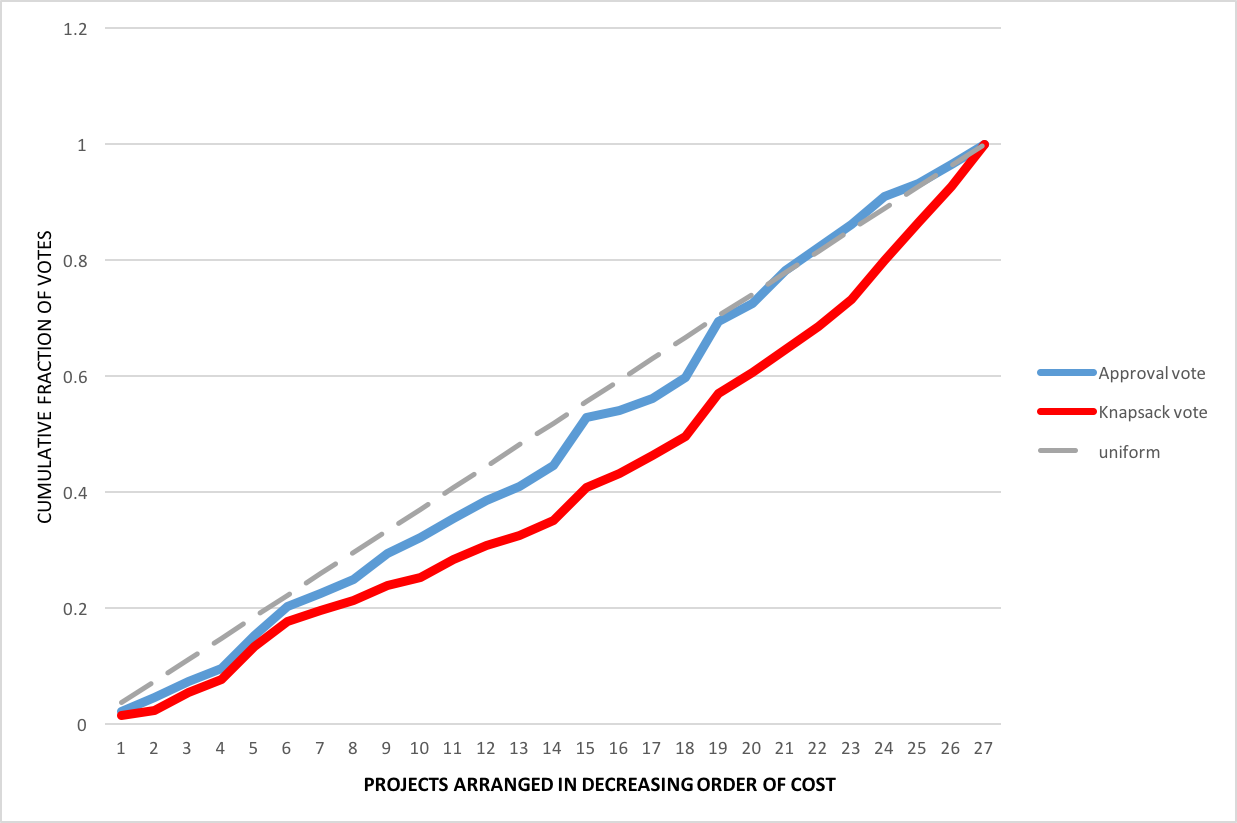}
  \caption{Cumulative fraction of total votes versus cost of projects: NYC District 8}
  \label{fig:nyc8_cumulative}
\end{minipage}
\end{figure}

\subsection{Cost consideration under Knapsack}\label{subsec:knapsackdata}
Our data suggests that there is a bias towards projects of larger costs in the K-approval method, as compared to the Knapsack method. We empirically verify this effect in two ways. First, we present data from Cambridge and NYC District 8, where a budget of \$600,000 and 6-approval, and \$1,000,000 and 5-approval, were used respectively. In Figures \ref{fig:cambridge_cumulative} and \ref{fig:nyc8_cumulative} we lay out the projects in descending order of cost, and plot the cumulative fraction of votes for projects above every cost threshold. We then compare K-approval and Knapsack against the uniform distribution. We see that this function for K-approval dominates that for 
Knapsack, which means that costlier projects are over-represented, thereby 
supporting our hypothesis. While this observation is not direct evidence that Knapsack Voting is better than K-approval, it does suggest a significant qualitative difference between the outcomes of the two methods.

\begin{table}[ht!]
\begin{minipage}{\textwidth}
	\centering
	    	\caption{Average cost of winning projects, as a fraction of the budget}
    		\label{tab:avgcost}
    	\begin{tabular}{ | r | r | r | }
    		\hline
      		& \emph{K-approval} & \emph{knapsack} \\ \hline
    		NYC District 5 &  0.18 & 0.14 \\ \hline
    		NYC District 8 & 0.20 & 0.12 \\ \hline
    		Boston 2015& 0.27 & 0.18 \\ \hline
    		Cambridge 2015 & 0.15 & 0.10 \\ \hline
    	\end{tabular}
	\end{minipage}
\end{table}

Second, we look at the average cost of the winning projects under each method. Table \ref{tab:avgcost} shows the average cost of the winning projects (normalized by the total budget) in each of those elections. On the average, across the three places, there is a reduction of about 30\% in the average cost of the winning projects. The above two observations clearly suggest that Knapsack Voting leads to voters' being more frugal while choosing which projects to vote for.

\subsection{Timing data}\label{subsec:timing}
The data from Tables \ref{tab:time_kapproval} and \ref{tab:time_knapsack} suggests that the Knapsack interface is not much more time consuming than the K-approval interface. Of course, since the Knapsack interface follows the official K-approval interface, and so the voters were familiar with the projects when they attempted the Knapsack vote. Even if the time taken by Knapsack were the sum of times for both in our experiments, we can see that it is very reasonable. This is borne out in the Boston 2016 election (see Table \ref{tab:time_knapsack}), where the number of projects was the same as in Boston 2015 and Knapsack was presented first as the official ballot.

Here are some more interesting observations on the ease of each voting method based on Tables \ref{tab:time_kapproval},\ref{tab:time_vfm}, \ref{tab:time_knapsack}: 
\begin{itemize}
\item we presented to each voter either Knapsack or Value-for-money comparisons with a 50\% chance in both Boston 2015 and Cambridge 2015, and the percentage of voters that completed the comparisons was at least a third greater than that of Knapsack;
\item in Boston 2015, the median time taken by the voters was $61$ seconds for K-approval vote (K=4 out of 10 projects on the ballot), $70$ seconds for Knapsack and $14$ seconds for value-for-money comparisons of 4 pairs. A similar trend is seen in the Cambridge experiment. 
\end{itemize}
All this suggests that Knapsack Vote and K-approval have comparable times. \emph{value-for-money} comparisons involve a smaller cognitive load, than the other voting methods, especially with a large number of projects as seen in Cambridge 2015.

\begin{table}[ht!]
\begin{minipage}{0.47\textwidth}
	\centering
	        \caption{Timing data for K-approval; N = number of projects on ballot, T = median completion time.}
	        \label{tab:time_kapproval}
    \begin{tabular}{ | r | r | r | r |}
    \hline
      City & K & N & T\\ \hline
  	Boston 2015 &  4 & 10 & 61s\\ \hline
    Cambridge 2015 & 6 & 23 & 213s\\ \hline
    \end{tabular}
\end{minipage} 
\hspace{10pt}
\begin{minipage}{0.47\textwidth}
	\centering
	    \caption{Timing data for value-for-money comparisons; n = number of comparisons, T = median completion time.}
    \label{tab:time_vfm}
    \begin{tabular}{ | r | r | r | r |}
    \hline
      City & n & T & \% of voters\\ \hline
  	Boston 2015&  4 & 14s & 10\\ \hline
    Cambridge 2015 & 4 & 53s & 40\\ \hline
    \end{tabular}
\end{minipage}
\end{table} 

\begin{table}[ht!]
	\centering
	     \caption{Timing data for Knapsack Voting; N = number of projects, T = median completion time.}
    \label{tab:time_knapsack}
    \begin{tabular}{ | r | r | r | r | r |}
    \hline
      City & Budget & N & T & \% of voters\\ \hline
  	Boston 2015 &  \$1,000,000 & 10 & 70s & 7\\ \hline
    Cambridge 2015 & \$600,000 & 23 & 115s & 30\\ \hline
    Boston 2016 &  \$1,000,000 & 10 & 86s & 100\\ \hline
    \end{tabular}
\end{table}

\subsection{Comparison of Knapsack and K-approval against value-for-money comparisons}\label{subsec:knapvskapp}
Using the data from random pairwise comparisons in the value-for-money experiment, and based on the Set Borda score (Definition \ref{def:setborda}) we calculate the average number of comparisons that agree and disagree with the winning sets as determined by the K-approval and Knapsack methods. 
\begin{table}[ht!]
	\centering
	    	\caption{Comparison of Knapsack and K-approval based on agreement with pairwise comparisons: Knapsack has a greater level of agreement.}
    		\label{tab:borda}
    	\begin{tabular}{ | r | r | r | r | r |r|}
    		\hline
    		& \multicolumn{2}{c|}{K-approval}   & \multicolumn{2}{c|}{Knapsack} \\ \hline
      		& \textbf{Agreement} & Standard error & \textbf{Agreement} & Standard error \\ \hline
    		Boston 2015&  \textbf{0.20} & 0.047 & \textbf{0.26} & 0.045\\ \hline
    		Cambridge 2015&  \textbf{0.30} & 0.022 & \textbf{0.40} & 0.021\\ \hline
    	\end{tabular}
\end{table}
If we picked a dollar allocated in the winning outcome, and a dollar not in it, both at random, the numbers in Table \ref{tab:borda} are a measure of the average fraction of votes that agree with the outcome minus those that disagree. We see that on this score, \emph{Knapsack} does better than the \emph{K-approval}, thereby indicating a higher level of agreement with the voters' preferences. We also note the standard error in the sample mean calculated above, assuming the sample mean is normally distributed. We see that the difference between the sample means of Knapsack and K-approval is greater than the standard error.

%In Boston 2016, we implemented Knapsack as the official ballot. By law, official ballots have to be made available in a paper format also. To overcome this difficulty, we made voters rank their top 4 projects in terms of value-for-money (Figure \ref{fig:paperballot}). This method yielded the same outcome as the digital Knapsack votes.

\section{Conclusions and ongoing work}
\emph{Knapsack voting} and Value-for-money comparisons are intuitive ways of eliciting voters' preferences for budgetary decisions. The Knapsack Vote admits interesting strategic properties: in particular, it is strategy-proof under a natural utility model that depends on the overlap (and equivalently the $\ell_1$ distance) between the outcome and the voters' true preferred allocations . It can also be extended to more complicated settings with revenues, deficits and surpluses. However, Knapsack Voting is not very practical when there are a large number of projects. While \emph{Value-for-money} comparisons have some drawbacks, they provide a way of eliciting the voters' preferences with a small cognitive load, especially in the case of large ballots. The fact that our schemes do better on many different measures finds support in the data we collected from participatory budgeting elections in various cities/municipalities.  All our schemes are amenable to implementation using interactive digital tools, thereby enhancing the ability of voters to make more informed decisions in \emph{participatory budgeting}. We have been able to make some initial progress along implementing Knapsack Voting as the official ballot, and we hope that this paper makes a strong case for its wider adoption in practice.

\paragraph{\textbf{Acknowledgements}}
This work is supported by the Army Research office (grant \# 116388), the Office of Naval Research (grant \# 11904718), and the Stanford Cyber Initiative.

We also thank Tim Roughgarden for useful discussions.

\bibliographystyle{plainnat}

\newpage
\appendix
\section*{APPENDIX}
\setcounter{section}{1}

\subsection{Extending results (approximately) to an Integral Model}\label{app:integral}
The integral version of Knapsack Voting is defined as follows:
\begin{definition} [Integral Knapsack Vote] \label{def:intknapsack}
\begin{enumerate}
\item Each voter $v \in V$ votes for a (integral) subset $S_v \in P$, such that it satisfies the budget constraint $\sum_{p \in Sv} c_p \leq B$.
\item The projects are arranged in decreasing order of number of their approval scores, which for any project $p$ is given by $\#\{v \in V : p \in S_v\}$.
\item The projects are chosen in this order till there is not enough budget to include the next project.
\end{enumerate}
\end{definition}
Our results, for example Theorem \ref{thm:EMDstratproof}, do not hold in the above integral model. The following example (similar ones can be given for Theorems \ref{thm:welfmax} and \ref{thm:parttrue}) shows us why:
\begin{example}
We have a budget of $5$ and three projects $a$,$b$ and $c$ of cost $2$,$2$ and $3$ respectively. Let's say that ties are broken in favor of projects with larger cost whenever possible. For a voter $v$, say that her favorite outcome is $b$ and $c$. Without $v$'s vote being tallied, assume $a$ and $b$ are winning --  with scores $15$, $11$ and $10$ for $a$,$b$ and $c$. Now if she votes for $b$ and $c$, the outcome is $a$ and $b$. But if she switches to $a$ and $c$, the outcome becomes $a$ and $c$ leading to higher utility (since she likes $c$ more than $b$).
\end{example} 

However, our results do hold in an \emph{approximately integral} model, i.e., with some assumptions.
\begin{definition}[Approximately Integral Knapsack Vote]
\begin{enumerate}
\item Each voter $v \in V$ votes for a (integral) subset $S_v \in P$, such that it satisfies the budget constraint $\sum_{p \in Sv} c_p \leq B$.
\item The projects are arranged in decreasing order of number of their approval scores, which for any project $p$ is given by $\#\{v \in V : p \in S_v\}$.
\item The projects are chosen in this order till there is not enough budget to include the next project. And the next one in line is fractionally implemented.
\end{enumerate}
\end{definition}

In other words, although the votes are constrained to be integral, \emph{approximately integral outcomes} are allowed -- the first project that gets the most votes and cannot be fully funded is allowed to be fractionally implemented, so that the entire budget is used up. When a project $j$ in $S_v$ (voter $v$'s favorite set) is funded fully, the utility derived is equal to $c_p$. If it is fractionally funded, then the utility is pro-rated appropriately.  It may or may not be realistic to partially fund projects (e.g., renovate only one floor of a library rather than the entire building), but partial funding and the requirement of spending all of the budget are both common in practice.

\paragraph{Note:} As mentioned in Section \ref{subsec:discutil}, the $\ell_1$ costs model doesn't make sense in a fully integral model (on account of violating free disposal). However, with approximately integral outcomes, both $\ell_1$ costs and Overlap utilities satisfy free disposal.

For Lemma \ref{lem:ell1overlap}, it is easy to see that we have the following analog in the approximately integral case.
\begin{corollary}\label{lem:ell1overlapfrac}
$|S_v \cap S^*| = \frac{B + |S_v|}{2} - \frac{1}{2} \sum_{p \in \p} |w^v_p - w^*_p|$.
\end{corollary}

Using the above, and by a simple modification of the proofs of Theorems \ref{thm:EMDstratproof}, \ref{thm:welfmax} and \ref{thm:parttrue}, we can see that they extend to the approximately integral model.
\begin{observation}
The results of Theorems \ref{thm:EMDstratproof}, \ref{thm:welfmax} and \ref{thm:parttrue} hold under the approximately integral model.
\end{observation}

The above observation means that our results hold approximately in the strictly integral model under Overlap utilities (see note above). First, let's consider Theorem \ref{thm:EMDstratproof}. For example, let's say we have a voter whose favorite (integral) set of projects is $S_v$. Then under the approximately integral model, voting for $S_v$ is a weakly dominant strategy for voter $v$. This implies that the corresponding approximately integral outcome, composed of some integral set $S^*$ and a part of project $p^*$, yields at least as much utility as any other approximately integral outcome that $v$ can obtain by voting some other integral set $S_v^\prime$. Consequently, we have under the integral model that the set $S^*$ yields at least as much utility as any other integral set (that $v$ can obtain by changing her vote) minus the value of the project $p^*$. A similar observation can be made for Theorems \ref{thm:welfmax} and \ref{thm:parttrue}.

\begin{observation}
The results of Theorems \ref{thm:EMDstratproof}, \ref{thm:welfmax} and \ref{thm:parttrue} hold approximately, up to the value of one project, under the integral model.
\end{observation}

We are drawing on ideas such as $\epsilon$-strategyproofness (strategy proof upto a utility value of $\epsilon$, and relaxations in combinatorial fair division such as envy-freeness upto one good \citep{caragiannis2016unreasonable,budish2011combinatorial,lipton2004approximately}.

\subsection{Proof of Theorem \ref{thm:budgbal}} \label{proof:budgbal}
Let $R_{-v}, S_{-v}$ and $\score_{-v}(.)$ be the outcome determined by the votes of everyone except $v$  using the Knapsack Voting rule (Equation \ref{eqn:knapsackbudgbal}). Assume that $Q_v \subseteq \rr$ and $T_v \subseteq \pp$ be a best response for voter $v$ such that $Q_v \neq R_v$ and $T_v \neq S_v$. Let the outcome after incorporating $Q_v,T_v$ be $\score(.)$ and $S$. 

We will discuss the case when $|Q_v| = |T_v| < |S_v| = |R_v|$. The other case of $|Q_v| = |T_v| \geq |S_v| = |R_v|$ follows analogously (equality is similar to the proof of Theorem \ref{thm:EMDstratproof}).

  Choose some $j \in S_v \setminus T_v$ such that if $k = D^p_t$ for some $p \in \p$, then $D^p_{t^\prime} \in T_v$ for all $t^\prime < t$. Such a $k$ exists because of consistency and the fact that $|T_v| < |S_v|$. Let $T_v^\prime \triangleq T_v \cup \{j\}$. Similarly, choose the $k \in R_v \setminus Q_v$ such that if $k = D^q_z$ for some $q \in \rr$, then $D^q_{z^\prime} \in T_v$ for all $z^\prime < z$ and define $Q_v^\prime \triangleq Q_v \cup \{k\}$. Let the outcome here be $\score^\prime(.)$ and $R^\prime, S^\prime$. We will show that $Q_v^\prime, T_v^\prime$ is also a best response for $v$. 
  
  If $R = R^\prime$ and $S=S^\prime$, then the utility is unchanged and we have nothing to prove. 

We have that $\score^\prime(j) = \score(j) + 1$, $\score^\prime(k) = \score(k) + 1$, and for all $l \in \pp \cup \rr \setminus \{j,k\}$, $\score^\prime(l) = \score(l)$. Note that the only change from $\score(.)$ to $\score^\prime(.)$ is that the score of $j$ and $k$ increases. So the outcomes must satisfy $R^\prime \supseteq R$ and $S^\prime \supseteq S$, and $|R^\prime| = |S^\prime| \leq |S| + 2 = |R| +2$. For any given tie-breaking rule, if $S^\prime \neq S$, then we must have either $j \in S^\prime \setminus S$ or $k \in R^\prime \setminus R$, or both, i.e., any change in outcome must involve either $j$ or $k$ moving from outside the winning set to within.

If $j \in S^\prime \setminus S$, then there is a corresponding sub-project that is added to $R$ to maintain the budget balance, say $m \in R^\prime \setminus R$.  The change in utility of voter $v$ is $\indicator(j \in S_v) - \indicator(m \notin R_v) = 1 - \indicator(m \notin R_v) \geq 0$. In other words, the change in utility from adding $j$ to $S$ is 0, and any $m$ to $R$ to maintain budget balance, cannot be negative.

Similarly, if $k \in R^\prime \setminus R$, then there is a corresponding $m^\prime \in S^\prime \setminus S$. And the change in utility is $-\indicator(k \notin R_v) + \indicator(m^\prime \in S_v) = 0 + \indicator(m^\prime \in S_v) \geq 0$.

By repeating this process, we reach a point at which we have a best response $Q_v,T_v$ equal in size to $R_v,S_v$ respectively. From here we do a procedure similar to the proof of Theorem \ref{thm:EMDstratproof} until we only have elements in $S_v$. In this entire process, we do not decrease the utility.

To prove that the outcome is welfare-maximizing, note that it is given by $(R^*,S^*)$, where $R^*$ and $S^*$ are both consistent and $|R^*| = |S^*|$, which maximizes the following:
\begin{align*}
 \sum_{i \in S} \score(i) + \sum_{j \in R} \score(j) =  \sum_{i \in S} \sum_{v \in \V} \indicator(i \in S_v) - \sum_{j \in S} \sum_{v \in \V} \indicator(j \notin R_v) \\
  =  \sum_{v \in \V} \sum_{i \in S} \indicator(i \in S_v) -  \sum_{v \in \V} \sum_{j \in R} \indicator(j \notin R_v)= \sum_{v \in V} |S_v \cap S| - |R \setminus R_v|,
\end{align*}
and the last quantity in the above is the welfare according to the Overlap Utility Model.
\subsection{Proof of Theorem \ref{thm:parttrue}}\label{proof:parttrue}
We will use the per-dollar approach (see definition in Section \ref{sec:knapsack}), i.e., each voter $i \in \V$ submits a vote $S_i \subseteq \pp$ such that $S_i$ is consistent and $|S_i| = B$.
Let us first reinterpret the notation from Section \ref{subsec:parttrue} per-dollar, and restate Theorem \ref{thm:parttrue} in technical terms.

Consider a focal voter $i$ responding to the votes of all others. Assume that she has full knowledge about the how the others voted in aggregate. If $S_{-i}$ denotes the cumulative votes of all voters except $i$, she knows  
$\W_{-i}$, the set of winners as determined by $S_{-i}$. Let
$\W(S_i,S_{-i})$ denote the set of winners if her vote $S_i$ is added. 

A \emph{best response} for $i$ is then defined as
a vote $S_i^\star$ that satisfies 
\begin{equation}
S_i^\star = \arg \max_{S_i \in A} \sum_{j \in \W(S_i,S_{-i})} \vp_{i,j},
\end{equation}
where $A \triangleq \{S: S \mbox{ is consistent},~|S| = B \}$, and $\vp_{i,j}$ is the utility of voter $i$ from sub-project $j \in \pp$.

With respect to voter $i$, we say a candidate $q \in \pp$ \emph{dominates} $p \in \pp$ if and only if 
\begin{itemize}
\item $q \in \W_{-i}$, and
\item $\vp_{i,q} > \vp_{i,p}$.
\end{itemize}
Let $ \Lambda_{i,p} \triangleq \{ j \in W_{-i}: \; \vp_{i,j} > \vp_{i,p} \} $ denote the collection of candidates that dominate $p$ with respect to voter $i$.

%Theorem \ref{thm:parttrue} can be restated as follows:
%
%If there is a best response $S_i^\star$ for $i$, such that $p \in S_i^\star$,
%then there exists another best response $S_i^{\star\star}$ such that 
%$ \Lambda_{i,p} \subseteq S_i^{\star\star}$. \\

Let's say that $p\in S_i^\star$ and $\Lambda_{i,p} \nsubseteq S_i^\star$. We will first claim that $\Lambda_{i,p} \subseteq \W(S_i^\star,S_{-i})$. 

Let $j \in \Lambda_{i,p}$.
We have the following two possible cases:
\begin{enumerate}
\item If $j \in S_i^\star$: since $\Lambda_{i,p} \subseteq W_{-i}$, we have $j \in S_i^\star \cap
\W_{-i}$, and consequently $j \in \W(S_i^\star,S_{-i})$.
\item Else, if $j \notin S_i^\star$: Assume $j\notin \W(S_i^\star,S_{-i})$. Then $S_i^\star$ cannot be a best
response, because by switching her vote from $S_i^\star$ to $\left(S_i^\star \setminus
\{p\}\right) \cup \{j\}$, $i$ can make $j$ win instead of $p$, and this
strictly increases her total utility (since $v_{i,j}>v_{i,p}$). 
\end{enumerate}
Therefore, in either of the above-mentioned cases, we are guaranteed that $ j \in \W(S_i^\star,S_{-i})$. 

Now, if $S_1 \triangleq S_i^\star \setminus \Lambda_{i,p} $, then $|S_1| = B - |S_i^\star \cap \Lambda_{i,p}|$. And since we have proved $  \Lambda_{i,p} \subseteq \W(S_i^\star,S_{-i})$, it follows that $|S_1 \cap \W(S_i^\star,S_{-i})| \leq B - |\Lambda_{i,p}|$. These two facts together imply the following:
\begin{eqnarray*}
|S_1 \setminus \W(S_i^\star,S_{-i})| = & |S_1| - |S_1 \cap \W(S_i^\star,S_{-i})| \\
\geq & \left(B - |S_i^\star \cap \Lambda_{i,p} \right) - \left( B - |\Lambda_{i,p}| \right) \\
\geq & |\Lambda_{i,p}| - |S_i^\star \cap \Lambda_{i,p}| \\
= & | \Lambda_{i,p} \setminus S_i^\star|.
\end{eqnarray*}
Hence, there exists $S_2 \subseteq S_1 \setminus \W(S_i^\star,S_{-i})$ such that
$|S_2|=|\Lambda_{i,p}  \setminus S_i^\star|$.

Let $S_i^{\star\star} = (S_i^\star \setminus S_2) \cup \Lambda_{i,p} $.
Clearly, $|S_i^{\star\star}|=B$ (since $S_2
\cap \Lambda_{i,p} =
\emptyset $), and so it is a valid vote for voter $i$.

Also $\W(S_i^{\star \star},S_{-i}) = \W(S_i^\star,S_{-i})$, since, we have replaced $S_2 \subseteq \W(S_i^\star,S_{-i})^c$ with $\Lambda_{i,p} \setminus S_i^\star
\subseteq \W(S_i^\star,S_{-i})$. Because $S_i^\star$ is a best response, so
is $S_i^{\star\star}$.

Note that \begin{equation}\label{eqn:respdiff}
S_2 = S_i^{\star} \setminus S_i^{\star\star} \subseteq {\W(S_i^\star,S_{-i})}^c
\end{equation}

Now if there is a $\ppp \in S_i^{\star\star}$ such that $\Lambda_{i,\ppp} \nsubseteq S_i^{\star\star}$, we can do a similar replacement procedure to define another best response $S_i^{\star\star\star} = (\sss \setminus S_2^\prime) \cup \Lambda_{i,\ppp}$, such that $S_2^\prime \subseteq {\W(\sss,S_{-i})}^c$ (from equation \ref{eqn:respdiff}).  Since $\Lambda_{i,p} \subseteq \W(\sss,S_{-i})$, this implies that $S_2^\prime \cap \Lambda_{i,p} = \emptyset$ and so $S_i^{\star\star\star}$ includes both $\Lambda_{i,p}$ and $\Lambda_{i,\ppp}$.
By a series of replacements, we have the best response as required by the theorem.
\qed

\end{document}